%% file: main.tex
\documentclass[11pt]{article}
\usepackage[numbers]{natbib}
\usepackage{xspace}
\usepackage[usenames,dvipsnames]{xcolor}
\usepackage{graphicx}
\usepackage{amsmath,amssymb,amsthm}
\usepackage[margin=1in]{geometry}
\usepackage[algoruled,noend]{algorithm2e}
\usepackage{pstricks}
\usepackage{comment}
\usepackage{caption}
\usepackage{subcaption}
\usepackage{dsfont}
\usepackage{tikz}
\usepackage{tikz-dependency}
\DeclareMathOperator*{\argmax}{arg\,max}

\usepackage{thmtools} 
\usepackage{thm-restate}

\usepackage[colorlinks=true,urlcolor=blue,linkcolor=RoyalBlue,citecolor=OliveGreen]{hyperref}

\newtheorem{theorem}{Theorem}

\newtheorem{claim}{Claim}
\newtheorem{lemma}{Lemma}
\newtheorem{corollary}{Corollary}
\newtheorem{proposition}{Proposition}

\usepackage{tikz}
\usetikzlibrary{patterns}
\usetikzlibrary{arrows,decorations.pathreplacing}

\newenvironment{talign}
 {\align}
 {\endalign}
\newenvironment{talign*}
 {\csname align*\endcsname}
 {\endalign}

\newcommand{\rev}[1]{\textsc{Rev} ( #1 )}
\newcommand{\brev}[1]{\textsc{BRev} ( #1 )}
\newcommand{\srev}[1]{\textsc{SRev} ( #1 )}
\newcommand{\secondprev}[1]{\textsc{RonenRev} ( #1 )}
\newcommand{\secondpricerev}[1]{\textsc{SecondPrice} ( #1 )}
\newcommand{\BICrev}[1]{\textsc{BICRev} ( #1 )}
\newcommand{\drev}[1]{\textsc{DSICRev} ( #1 )}
\newcommand{\arev}[1]{\textsc{AdditiveRev} ( #1 )}

\usepackage{enumitem}

\newcommand{\realD}{\mathcal{D}}
\newcommand{\wiggleD}{\hat{\realD}}
\newcommand{\trig}{\text{trig}}
\newcommand{\gap}{\text{gap}}

\allowdisplaybreaks

\begin{document}
\title{Smoothed Analysis of Multi-Item Auctions with Correlated Values}
\date{}

\author{
Christos-Alexandros Psomas
\thanks{Computer Science Department, Carnegie Mellon University, cpsomas@cs.cmu.edu.} \and 
Ariel Schvartzman
\thanks{Department of Computer Science, Princeton University, acohenca@cs.princeton.edu.} \and
S. Matthew Weinberg
\thanks{Department of Computer Science, Princeton University, smweinberg@princeton.edu. Supported by NSF CCF-1717899}
}

\maketitle

\begin{abstract}
Consider a seller with $m$ heterogeneous items for sale to a single additive buyer whose values for the items are arbitrarily correlated. It was previously shown that, in such settings, distributions exist for which the seller's optimal revenue is infinite, but the best ``simple'' mechanism achieves revenue at most one (\citet{BriestCKW15, HartN13}), even when $m=2$. This result has long served as a cautionary tale discouraging the study of multi-item auctions without some notion of ``independent items''.

In this work we initiate a smoothed analysis of such multi-item auction settings. We consider a buyer whose item values are drawn from an arbitrarily correlated multi-dimensional distribution then randomly perturbed with magnitude $\delta$ under several natural perturbation models. On one hand, we prove that the~\cite{BriestCKW15, HartN13} construction is surprisingly robust to certain natural perturbations of this form, and the infinite gap remains. 

On the other hand, we provide a smoothed model such that the approximation guarantee of simple mechanisms is smoothed-finite. We show that when the perturbation has magnitude $\delta$, pricing only the grand bundle guarantees an $O(1/\delta)$-approximation to the optimal revenue. That is, no matter the (worst-case) initially correlated distribution, these tiny perturbations suffice to bring the gap down from infinite to finite. We further show that the same guarantees hold when $n$ buyers have values drawn from an arbitrarily correlated $mn$-dimensional distribution (without any dependence on $n$).

Taken together, these analyses further pin down key properties of correlated distributions that result in large gaps between simplicity and optimality.
\end{abstract}

\thispagestyle{empty}
\newpage
\setcounter{page}{1}

\input{intro}
\input{notation}

\input{rect_lb_new}

\section{Square-Shift Upper Bounds}\label{sec:square}
\input{angle_shift_newer}

\input{state_results}
\input{box_beyond}
\input{general_case}
\input{extensions}
\input{conclusion}

\newpage
\begin{appendix}
\input{table_appendix}
\input{lb_appendix}

\input{ub_appendix}

\input{silly_model}
\input{badmodels}
\end{appendix}

\bibliographystyle{plainnat}
\bibliography{refs}

\end{document}

%% file: intro.tex
\section{Introduction}\label{sec:intro}
How should a revenue-maximizing seller sell heterogeneous goods to interested buyers? This problem has been extensively studied by economists and computer scientists alike, from a variety of perspectives. One major highlight from these works is numerous impossibility results, essentially proving that one cannot hope to find a mechanism that is simple, yet optimal in general settings~\cite{BriestCKW15, HartN13, HartR15, DaskalakisDT14}. A major highlight on the flip side are numerous approximation results, where simple mechanisms are now known to be approximately optimal in quite general (but still structured) settings~\cite{ChawlaHK07, HartlineR09, ChawlaHMS10, HartN12, BabaioffILW14}. 

Perhaps the clearest example demonstrating the interaction of these two lines of work is the following. Consider a single additive buyer whose values for the two items are drawn jointly from a two-dimensional distribution $\realD$\footnote{That is, $(v_1,v_2) \sim \realD$, and the buyer's value for receiving both items is $v_1 + v_2$ (and receiving just item $i$ is $v_i$).}. Then there exist correlated distributions $\realD$ such that the revenue-optimal mechanism achieves infinite revenue, yet the revenue of any simple mechanism is at most $1$~\cite{BriestCKW15, HartN13}. Without getting into details of exactly what ``simple'' means\footnote{Formally, these results show lower bounds on a measure of simplicity called the menu-size complexity.}, this impossibility result rules out any hope of simple mechanisms that are even approximately optimal for all two-dimensional $\realD$. Still, a fantastic complementary result shows that if the two item values are drawn independently, selling separately (post a price $p_i$ for each item and let the buyer purchase a single item $i$ for $p_i$, or both for $p_1 + p_2$) guarantees a $2$-approximation~\cite{HartN12}. So while the impossibility results for arbitrary $\realD$ are quite strong, compelling positive results still exist under natural assumptions.

This avenue turned out to be quite productive: A long line of work recently culminated in a simple and approximately optimal mechanism for any number of buyers with subadditive valuations over any number of independent items~\cite{ChawlaHK07, ChawlaHMS10, ChawlaMS15, KleinbergW12, HartN12, LiY13, BabaioffILW14, Yao15, RubinsteinW15, CaiDW16, ChawlaM16, CaiZ17}. Modeling assumptions aside, this body of works constitutes a major contribution to the theory of optimal auction design.

The impact of these works notwithstanding, one key direction is left largely unaddressed: even as these works generalized in various directions (arbitrary feasibility constraints, combinatorial valuations, etc.), the ``independent items'' assumption remained. Even the few works that pose models of limited correlation have some underlying notion of independence (e.g. there are ``independent features,'' and item values depend linearly on features)~\cite{ChawlaMS15, BateniDHS15}. While ``independent items'' is a perfectly natural assumption (and we have greatly deepened our understanding of mechanism design under it), it was never intended to be ubiquitous in all future works. This is especially true due to the nature of the motivating impossibility result: the $\realD$ witnessing these impossibilities is so carefully crafted (we overview the construction in Section~\ref{sec:rectangle}) that it is far removed from any ``real-world'' motivation. That is, the constructions provided in~\cite{BriestCKW15, HartN13} require carefully building $\realD$ by perfectly placing the infinitely many points in its support just so, and even a tiny deviation in the construction would cause the entire argument to collapse. The thoughtful reader may at this point be thinking: if this construction is so fragile that even a tiny deviation breaks it, perhaps a smoothed analysis might prove insightful. Indeed, this is the focus of this paper. 

So, what might a ``smoothed'' distribution look like? Given an arbitrarily correlated distribution $\realD$, its smoothing $\wiggleD$ first draws a valuation vector $\mathbf{v}$ from $\realD$, and then randomly perturbs $\mathbf{v}$ to $\hat{\mathbf{v}}$ (where the size of the perturbation is parameterized by some $\delta > 0$). This makes sense for the same reason that it makes sense in all other applications of smoothed analysis: these distributions come from somewhere (e.g. past data), inevitably in the presence of tiny noise.

\subsection{What Makes a Good Smoothed Model?}

Once we've decided to do smoothed analysis, we also need to pick a model of random perturbation. We'll be interested in balancing relevance (e.g. how natural a model is) with transparent analysis (e.g. what insights can we derive?). To help illustrate this point, consider the following perturbation proposal: assume that $\realD$ is supported on $[0,v_{\max}]^2$,\footnote{The constructions of~\cite{BriestCKW15, HartN13} work subject to this, but one has to replace ``infinite gap'' with ``unbounded gap''.} and perturb $\mathbf{v}$ by adding a uniformly random vector from $[0,\delta v_{\max}]^2$. 

This is indeed a natural starting point. Unfortunately, the model lends itself to a trivial analysis. Specifically, we can first conclude that the optimal revenue for $\realD$ (respectively $\wiggleD$) is at most $2v_{\max}$ (respectively $2v_{\max}(1+\delta)$). Moreover, for $\wiggleD$, we can set a price of $\delta v_{\max}$ on the grand bundle (that is, allow the buyer to pay $\delta v_{\max}$ to receive all items or $0$ to receive nothing), which sells with probability at least $1/2$. Therefore, the revenue achieved by bundling all items together (henceforth, $\brev{\wiggleD}$) is at least $\delta v_{\max}/2$, immediately guaranteeing an $O(1/\delta)$-approximation. 
For the sake of completeness, we include an improved analysis in Appendix~\ref{app:silly} showing that $\brev{\wiggleD}$ is in fact a $O(\ln(1/\delta))$-approximation in this model (and that this is nearly tight). So, while technically there's a ``positive result'' here, we simply don't learn much from this exercise. For readers especially interested in the modeling aspect, Appendix~\ref{app:badsmoothed} overviews similar issues with alternative models.

From here, we consider two natural modifications. We call the first \emph{Rectangle-Shift}, which essentially replaces the additive shift of the previous model with a multiplicative shift. In this model, values $(v,w)$ are perturbed to $(v+\varepsilon_1v, w+\varepsilon_2w)$, where each $\varepsilon_i$ is drawn independently from $U[0,\delta]$. Our first main result proves that in fact the infinite gap persists in this model! That is, for all $\delta < 1/2$, there exists a correlated $\realD$ such that the corresponding $\wiggleD$ satisfies $\rev{\wiggleD} = \infty$ but $\brev{\wiggleD} \in O(1)$ (which implies that any ``simple'' mechanism achieves finite revenue as well). 

\begin{restatable}{thm}{mainLB}
\label{thm:LBRect}
For all $\delta < 1/2$, there exists a bivariate distribution $\realD$  such that for its corresponding perturbed distribution $\wiggleD$ in the Rectangle-Shift model, $\brev{\wiggleD} \in O(1)$ and $\rev{\wiggleD} = \infty$.
\end{restatable}

We call our second model \emph{Square-Shift}, which essentially replaces the scale $v_{\max}$ in the initial model with $\max\{v,w\}$. In this model, values $(v,w)$ are perturbed to $(v + \varepsilon_1\max\{v,w\}, w + \varepsilon_2 \max\{v,w\})$, where each $\varepsilon_i$ is drawn independently from $U[0,\delta]$. Our second main result shows that $\brev{\wiggleD} \in \Omega(\delta \cdot \rev{\wiggleD})$ (and this is nearly tight, see Appendix~\ref{app:lower bound square two}).
\begin{restatable}{thm}{mainUB}
\label{thm: n=1,m=2 box}
In the Square-Shift model, for a single additive buyer and $2$ items 
\[ \rev{\wiggleD} \leq \frac{\sqrt{2} \pi (1+\delta)^3 \log{(1+\delta)}}{\delta^2} \brev{\wiggleD} \in O\left(\frac{1}{\delta} \cdot\brev{\wiggleD}\right). \]
\end{restatable}


Before continuing, we briefly share the distinguishing feature causing Rectangle-Shift to admit an infinite gap, yet Square-Shift to admit finite approximations: the \emph{angle} by which a buyer's valuation vector may be perturbed. Observe that in the Rectangle-Shift model, valuation vectors extremely close to an axis remain extremely close to that axis after perturbation. This fact turns out to be crucial in enabling a lower bound construction. On the other hand, valuation vectors in the Square-Shift model are likely to have their angle non-trivially perturbed. This property turns out to be crucial in establishing our approximation guarantee. Of course, both results require much more than this simple observation, but this property is the distinguishing factor causing results in the two models to diverge.

We conclude with two additional technical observations. First, note that in both models $\brev{\realD} \approx \brev{\wiggleD}$ (unlike the initial additive model). This means that the entirety of our analysis rests on studying $\rev{\wiggleD}$, and determining whether the perturbations guarantee that it's finite. Second, note that in both models $\wiggleD$ \emph{stochastically dominates} $\realD$. Therefore, when we prove in Theorem~\ref{thm: n=1,m=2 box} that $\rev{\wiggleD}$ is bounded, this is \emph{not} because we perturbed the buyer into valuing the items less, but really because this perturbation negates whatever bizarre properties of the original $\realD$ led to $\rev{\realD} =\infty$. This serves as another example of revenue non-monotonicity~\cite{HartR15}.\\

\noindent\textbf{Brief Discussion of Models.} The refuted additive-noise model perhaps seems most natural from a mathematical perspective (indeed, it is the most similar to models used in prior smoothed analysis). From an economic perspective, this is somewhat natural as well (each consumer in the population makes errors on the scale of, e.g., \$10, while no consumer values any item more than \$1 million). The ``Rectangle-Shift'' model is also uncontroversially natural from both perspectives: a consumer's value for item $j$ is inaccurately measured proportional to her value for item $j$. The ``Square-Shift'' model requires a touch more thought. From a mathematical perspective, it is simply replacing the universal \$1 million scale for all consumers in the population with a scale proportional to that consumer's value for the items. From an economic perspective, it may initially seem uncompelling that a high value for one item may cause larger error in estimating the value of the other item. However, numerous works in the behavioral economics literature suggest that bidders indeed value items differently in the presence or absence of more/fewer valuable items~\cite{AdavalWyer, NunesBoatwright04, KWYA08, MKA13}. For example, in an experiment of~\cite{NunesBoatwright04}, the authors observed that a buyer's willingness to pay for a cheap item (a CD) was (statistically significantly) higher when a second item (a sweater) was being sold at 80\$ than when it was being sold at 10\$. 



\subsection{Extensions: Many Buyers and Many Items}
After resolving the single-buyer, two-item case, the natural question is whether a similar analysis extends to multiple buyers or multiple items. For the Rectangle-Shift model, the impossibility results extend to the multi-item case, so we focus on the Square-Shift model. For $n$ buyers and two items, we consider an arbitrarily correlated $2n$-dimensional distribution $\realD$ (denoting $n$ buyers' values for two items). Our perturbation then draws $\varepsilon_{ij}$ independently and uniformly from $[0,\delta]$ for all buyers $i$, and $j \in \{1,2\}$, and maps $v_{ij}$ to $v_{ij} + \varepsilon_{ij} \max_{j' \in \{1,2\}} v_{ij'}$. In other words, while the buyers' values are correlated, the perturbations are independent (and only depend on that buyer's values). 

In this model, we're again able to prove that $\secondprev{\wiggleD} \in \Omega(\delta  \cdot \drev{\wiggleD} )$ (and this is again nearly tight). Here, $\drev{\wiggleD}$ denotes the revenue of the optimal dominant-strategy truthful mechanism\footnote{That is, it is in each buyer's interest to report their true value, no matter the behavior of the other buyers. Contrast this with Bayesian truthful, where it is in each buyer's interest to report their true value as long as the other buyers do so as well. Note that we cannot hope to replace $\drev{\wiggleD}$ with $\BICrev{\wiggleD}$ here due to~\cite{FuLLT17} - see Section~\ref{sec:notation} for further discussion.}, and $\secondprev{\wiggleD}$ denotes the revenue achievable by running Ronen's simple single-item auction~\cite{Ronen01} for the grand bundle, i.e. treat the grand bundle as a single-item auction and run Ronen's auction. See Section~\ref{sec:notation} for the definition of Ronen's auction --- it is a second-price auction with reserve, but the reserve depends on the other buyers' bids. If $\realD$ happens to be independent across buyers, but still correlated across items, we further get $\secondprev{\wiggleD} \in \Omega(\delta \cdot \BICrev{\wiggleD})$, where $\BICrev{\wiggleD}$ is the revenue of the optimal BIC mechanism. Note that the guarantee does not depend on $n$. 

Finally, we consider the single-buyer, multi-item case. The extension of the Square-Shift model is the obvious one: perturb each value independently by $U([0,\delta])\cdot \max_j v_j$. Here, we show that our techniques extend, but give $\brev{\wiggleD} \in \tilde{\Omega}(\delta^m\cdot \rev{\wiggleD})$. 
Furthermore, the exponential dependence on $m$ is unavoidable: a simple counterexample provides a $\realD$ such that for $\delta \in O(1)$, $\brev{\wiggleD} < \frac{m}{2^m} \cdot \rev{\wiggleD}$ (Appendix~\ref{app:lower bound square}). Our analysis  extends to the $m$-item $n$-buyer case, again obtaining an approximation guarantee exponential in $m$. 

\begin{restatable}{thm}{singlebidUB}
\label{thm: n=1,general m box}
In the Square-Shift model, for a single additive buyer and $m$ items
$$\rev{\wiggleD} \leq \left( \sqrt{\frac{\pi e}{2}} \frac{(1+\delta) }{\delta} \right)^m (1+\delta) \log(1+\delta)  m \sqrt{m} \cdot \brev{\wiggleD}. $$
\end{restatable}

\begin{restatable}{thm}{multiUB}
\label{thm:multi buyer multi item UB}
For the $n$ additive buyer, $m$ item, Square-Shift model,
\[ \rev{\wiggleD} \leq 4 \left( \sqrt{\frac{\pi e}{2}} \frac{(1+\delta) }{\delta} \right)^m (1+\delta) \log(1+\delta)  m \sqrt{m} \cdot \secondprev{\wiggleD}, \]
where $\rev{\wiggleD} = \drev{\wiggleD}$ for correlated buyers, and $\rev{\wiggleD} = \BICrev{\wiggleD}$ for independent buyers.
\end{restatable}

Our analysis further extends to settings beyond those where the buyers' valuation is additive across items. For example, all of our positive results hold for the general class of valuations which are \emph{additive subject to downwards-closed constraints}, at the cost of an additional factor of $m$; we discuss this in detail in Section~\ref{sec:extensions}. 

\subsection{Related Work and Roadmap}\label{subsec:related}
The present paper is the first to explore smoothed analysis in auction design. Other hybrid worst-case/average-case guarantees have been studied, for instance, in the digital goods setting (e.g.~\cite{GoldbergHKSW06, GoldbergH03, ChenGL14}), where at a high level, auctions compete instance-by-instance against any auction that could potentially be optimal in the average-case. Work of~\citet{carroll2017robustness} on robust mechanism design is also thematically related, and argues that simple mechanisms are optimal if the auctioneer wishes to obtain a worst-case guarantee against all value distributions consistent with given marginals (but neither directly nor indirectly addresses the concepts in the constructions of~\cite{BriestCKW15, HartN13}).  

There is a growing body of related literature on multi-item auction design, largely proving  impossibility results for arbitrarily correlated distributions (e.g.~\cite{BriestCKW15, HartN13}), or approximation results for distributions with ``independent items'' (e.g.~\cite{ChawlaHK07, ChawlaHMS10, ChawlaMS15, KleinbergW12, HartN12, LiY13, BabaioffILW14, Yao15, RubinsteinW15, CaiDW16, ChawlaM16, CaiZ17}). Limited work exists on models with limited correlation, but such models make use of the same tools, essentially replacing ``independent items'' with ``independent features'' and items that are linear combinations of features~\cite{ChawlaMS15, BateniDHS15}. The present paper is the first tractable multi-item model of limited correlation that doesn't rely on these tools.

Smoothed analysis is a popular framework for analyzing algorithms on ``real-world worst-case'' inputs. Smoothed analysis most commonly refers to smoothed computational complexity (e.g. an algorithm might run in exponential time in the worst case, but in polynomial time if the worst-case inputs are randomly perturbed), and has been an extremely influential paradigm~\cite{SpielmanT09}. For instance, the Simplex Method for solving LPs is known to take exponential time in the worst case, but has smoothed-polynomial computational complexity~\cite{SpielmanT04}. More similar in spirit to the present paper is prior work that considers the smoothed competitive ratio of online algorithms~\cite{BechettiLMSV06} or smoothed approximation ratio of mechanisms~\cite{DengGZ17}. The motivation for considering smoothed analysis in these works is, of course, similar to ours, but there is no similarity in techniques: the process of proving smoothed guarantees is a domain-specific process. 


\noindent \textbf{Roadmap:} Section~\ref{sec:notation} poses our model and some preliminaries. Section~\ref{sec:rectangle} presents our lower bound for the Rectangle-Shift model. In Section~\ref{sec:square} we present our results in the Square-Shift model, including its extension to multiple buyers and multiple items. Appendix~\ref{app:table} contains a detailed summary of our results in the Square-Shift model, as well as comparisons with known results for independent items. Extensions beyond additive buyers are presented in Section~\ref{sec:extensions}. Nearly tight bounds for the additive noise model can be found in Appendix~\ref{app:silly}. 

%% file: notation.tex
\section{Preliminaries}
\label{sec:notation}

Throughout this paper we study auctions with $n$ buyers who are bidding for $m \geq 2$ items. Each buyer $i$ draws its valuation vector $\mathbf{v}_i$ for the items from an $m$-dimensional distribution $\realD_i$ with density $f_i$, and we refer to $\realD$ as the joint $nm$-dimensional distribution. The buyer has value $v_{i,j}$ for the $j$-th item, with marginal distribution $\realD_{i,j}$. Our results do not assume any kind of independence between buyers or items. That is, $\realD_{i,j}$ can be correlated with $\realD_{i,j'}$ (same buyer different item), as well as $\realD_{i',k}$ (different buyer and possibly different item). For now, all buyers considered are additive, quasi-linear, and risk-neutral (or expected utility maximizers). See Section~\ref{sec:extensions} for extensions to non-additive valuation functions.

\paragraph{Mechanisms and Benchmarks.} 
As with many results in approximate mechanism design we compare the optimal revenue attainable to benchmarks that result from simple mechanisms. For a single buyer, let $\rev{\realD}$ be the optimal revenue that can be extracted by a truthful mechanism when valuations are sampled according to $\realD$. Let $\brev{\realD}$ be the optimal revenue that can be attained by selling the grand bundle as a single item --- that is to say, the mechanism posts a single price and allocates all the items if the bidder is willing to pay that price, otherwise it charges nothing and allocates nothing. For multiple buyers, let $\BICrev{\realD}$ denote the optimal revenue that can be extracted by a BIC mechanism (see Appendix~\ref{app:prelim} for a formal definition for correlated buyers), $\drev{\realD}$ denote the optimal revenue that can be extracted by a DSIC mechanism, and $\secondprev{\realD}$ denote the revenue extracted by Ronen's single-item auction (treating the grand bundle as a single item). Ronen's auction will always award the item to the highest bidder (or no one), and the price charged is the maximum of the second-highest bid and a per-buyer reserve (that depends on the other buyers' bids). Ronen shows that for any single-item setting, $2 \secondprev{\realD} \geq \drev{\realD}$.~\citet{FuLLT17} show that $\secondprev{\realD}$ does not generally guarantee any finite approximation to $\BICrev{\realD}$, but that under some assumptions it attains a 5-approximation.\footnote{All of these claims assume that the mechanism is required to be ex-post individually rational, due to seminal work of~\citet{CremerM88}: see Appendix~\ref{app:prelim} for further discussion.}

\paragraph{Models.} 
We consider two different smoothing models and refer to the resulting distributions as $\wiggleD_i$. The magnitude of the perturbation depends on a parameter $\delta$. We write $\hat{f}_i$ for the density of $\wiggleD_i$. $R(\hat{\mathbf{x}},\delta) = \{ \mathbf{x} |  \mathbf{x} \text{ can get mapped to } \hat{\mathbf{x}} \}$ is the set of points that could map to $\hat{\mathbf{x}}$ under a certain model with parameter $\delta$; the model will always be clear from the context and is omitted from notation. We also drop $\delta$ when clear from the context, and write $R^{-1}(\mathbf{x},\delta)$ for the set of points that $\mathbf{x}$ maps to.
The models we consider are (see Figures~\ref{fig:sub3},~\ref{fig:sub2} in Appendix~\ref{app:table} for illustrations): 

\begin{itemize}[leftmargin=*]
\item \textbf{Rectangle-Shift:} buyer $i$'s value $\mathbf{v}_i$ is replaced by $\boldsymbol{\hat{v}}_i$ with $\hat{v}_{ij} = v_{ij}+\delta_j v_{ij}$ where each $\delta_j$ is sampled independently from $U[0, \delta]$.  Intuitively, this spreads the mass $f(\mathbf{v})$ uniformly on the $m$-parallelepiped with side lengths $\delta \mathbf{v}$ and $\mathbf{v}$ as its smallest vertex. 
\item \textbf{Square-Shift:} buyer $i$'s value $\mathbf{v}_i$ is replaced by $\boldsymbol{\hat{v}}_i$ with $\hat{v}_{ij} = v_{ij}+\delta_j \max_k v_{ik}$ where each $\delta_j$ is sampled independently from $U[0, \delta]$. Intuitively, this spreads the mass $f(\mathbf{v})$ uniformly on the $m$-dimensional cube with side length $\delta \max_k v_{ik}$ and $\mathbf{v}$ as its smallest vertex. 
\end{itemize}

As mentioned earlier, it is known that even for the single buyer, two items case, the gap between $\rev{\realD}$ and $\brev{\realD}$ is unbounded when there is correlation between the items~\cite{BriestCKW15, HartN13}. In this paper we compare $\rev{\wiggleD}$ and $\brev{\wiggleD}$ as a function of $\delta$. 

%% file: rect_lb_new.tex
\section{Persistence of Infinite Gap in the Rectangle-Shift Model}\label{sec:rectangle}

In this section we first overview the key ideas from~\cite{BriestCKW15, HartN13}, and then present our construction witnessing an infinite gap in the Rectangle-Shift model. Proofs missing from this section can be found in Appendix~\ref{app: rectangle lower bounds}. Recall this section's main result:
\mainLB*




The key insight behind the~\cite{BriestCKW15,HartN13} construction is the following: assume that you have $k$ points in the positive orthant of the unit hypercube, ${\mathbf{x}}_1,\ldots, {\mathbf{x}}_k$, such that ${\mathbf{x}}_i \cdot {\mathbf{x}}_j < {\mathbf{x}}_i \cdot {\mathbf{x}}_i -\varepsilon$ for all $i, j$. If the valuation vectors in the support of our distribution are exactly ${\mathbf{x}}_1,\ldots, \mathbf{x}_k$, we can design a mechanism such that when the bidder reports a valuation of ${\mathbf{x}}_i$, the mechanism offers a randomized allocation (or \emph{lottery}) of ${\mathbf{x}}_i$ for price $\varepsilon$. When the buyer has valuation vector ${\mathbf{x}}_i$, they get utility ${\mathbf{x}}_i \cdot {\mathbf{x}}_i - \varepsilon$ for the lottery designed for them, which is greater than the utility of reporting some other valuation ${\mathbf{x}}_j$, ${\mathbf{x}}_i \cdot {\mathbf{x}}_j$. Therefore, their utility for purchasing the lottery tailored for them exceeds their \emph{value} for any other lottery (and therefore their utility as well). This guarantees that the buyer will always purchase the lottery designed for them. Observe that simply selling the grand bundle at price $\varepsilon$ achieves just as much revenue, so this idea is just the first building block.

From here,~\cite{BriestCKW15, HartN13} modify this construction to increase $\textsc{Rev}$ while keeping $\textsc{BRev}$ small. The second insight is that if the valuation vectors in the support are actually $2^i {\mathbf{x}}_i$, then we can offer the lottery ${\mathbf{x}}_i$ for price $\varepsilon \cdot 2^i$ and the buyer will still always prefer their tailored lottery to any other one, because their utility will be $2^i( \mathbf{x}_i \cdot \mathbf{x}_i - \varepsilon) > 2^i\mathbf{x}_i \cdot \mathbf{x}_j$, which is their value for the lottery $\mathbf{x}_j$. Therefore, if the buyer has value $2^i \cdot \mathbf{x}_i$ with probability $2^{-i}$, \textsc{Rev} can be quite large, while \textsc{BRev} will remain small. Of course, there are still some details left to work out, such as exactly how to analyze \textsc{Rev} and \textsc{BRev}, and we've also simplified the key ideas at the cost of technical accuracy, but these are the main ideas we'll borrow from previous work: the point is that these constructions are packing points inside the unit hypercube with small pairwise dot products. Our construction is a ``perturbation robust'' version of that of~\cite{HartN13}, as we have to show not only that ${\mathbf{x}}_i$ will prefer the intended option, but that any smoothing of ${\mathbf{x}}_i$ will prefer it as well. 

The intuition of the preceding paragraphs is captured by Lemma~\ref{lemma:lb auction} below. For a set of points $\mathbf{x}_1, \mathbf{x}_2, \dots \in [0,1]^2$ let $\text{gap}^i_{\delta} = \min_{j<i, \hat{\mathbf{x}}_i \in R^{-1}(\mathbf{x}_i,\delta)} \hat{\mathbf{x}}_i \cdot \left( \mathbf{x}_i - \mathbf{x}_j \right)$. That is, $\text{gap}^i_{\delta}$ roughly captures the minimum angle between any smoothing of $\mathbf{x}_i$ and any $\mathbf{x}_j$, $j < i$. First, we show that given any such set of points, there exists a distribution $\realD$ such that $\brev{\wiggleD}$ is at most a constant, but there exists an auction that extracts revenue $\sum_{i=1}^{\infty} \text{gap}^i_\delta$. Roughly speaking, $\realD$ has a type $\mathbf{v}_i$ for every point $\mathbf{x}_i$, such that all post-perturbation types $\hat{\mathbf{v}}_i$ prefer randomized allocation $\mathbf{x}_i$ at price $q_i = \gap^i_\delta$ over any other (allocation, price) pair $(\mathbf{x}_j,q_j)$. The proof of Lemma~\ref{lemma:lb auction} is similar to Proposition~9 of~\cite{HartN13}\footnote{This is clearer when compared to Proposition~$5.1$ in the arXiv version of~\cite{HartN13} uploaded April 22nd, 2013.} after adjusting for our perturbations.

\begin{lemma}
\label{lemma:lb auction}
In the Rectangle-Shift model, for all $\delta < 1/2$, and sequences of points $\mathbf{x}_1, \mathbf{x}_2,... \in [0,1]^2$, there exists a bivariate distribution $\realD$, such that for its corresponding perturbed distribution $\wiggleD$ it holds that $\brev{\wiggleD} \in O(1)$ and $\rev{\wiggleD} \geq \sum_{i=1}^{\infty} \text{gap}^i_\delta$.
\end{lemma}


Once we have Lemma~\ref{lemma:lb auction}, the crux of the proof becomes constructing a set of points such that $\sum_{i=1}^{\infty} \text{gap}^i_\delta$ diverges. We indeed follow an outline similar to the construction of~\cite{HartN13}, but the requirement that any \emph{smoothing} of $\mathbf{x}_i$ has sufficiently large angle with all $\mathbf{x}_j$ (versus just $\mathbf{x}_i$ itself) poses additional challenges. We include a fairly detailed overview below, with proofs of the key steps. 


\begin{lemma}\label{lemma:points}
For all $\delta < 1/2$, there exist a countably infinite set of points $\mathbf{x}_1, \mathbf{x}_2, \dots\in [0,1]^2$, such that $\sum_i \text{gap}^i_\delta$ diverges.
\end{lemma}
\begin{proof}
The points will lie on a sequence of spherical shells restricted to the non-negative orthant, with increasing radii. Specifically, shell $N$ will contain a consecutive subsequence of points, all with the same magnitude, $\ell_N$. We first make sure that the shells are sufficiently far apart from each other (i.e. that $\ell_N - \ell_{N-1}$ is sufficiently large) so that the distance between $\mathbf{x}_i$ on shell $N$ is large when compared to any $\mathbf{x}_j$ on shell $< N$. Specifically, we define $\ell_N = \frac{1}{\alpha}\sum_{k=2}^N c_k$, where $c_k = \frac{1}{k\log^2 k}$, and $\alpha = \sum_{k=2}^\infty c_k$. Note crucially that $\alpha$ is finite (by the integral test), and therefore $\ell_N$ is well-defined and bounded above by $1$ for all $N$. 

Now for each shell $N$, we'll place points with magnitude $\ell_N$ starting from angle $\pi/2$ and going down to angle $0$ such that \emph{within each shell}, the points are sufficiently far apart from each other. Specifically, we'll place points so that the angles form a geometric progression: the $k^{th}$ point on shell $N$ will have angle $\frac{\pi}{2}\cdot (1-3\delta)^{k-1}$, and we'll stop placing points on shell $N$ once the angle drops below $\sqrt{c_N}$ (we'll determine later exactly what the maximum $k$ is as a function of $N$).


From here we have two tasks of analysis. First, we must figure out what is $\gap^i_\delta$ for some point $\mathbf{x}_i$ on shell $N$. Second, we must figure out how many points we can pack in each shell with the above construction. We begin by analyzing the gap for a point on shell $N$ with the following lemma.

\begin{lemma}\label{lem:gap}
Let $\mathbf{x}_i$ be a point on shell $N$. Then $\gap^i_\delta \in \Omega(c_N)$ (recall $c_N = \frac{1}{N\log^2 N}$). 
\end{lemma}
\begin{proof}
In order to evaluate $\gap^i_\delta$ for a point in shell $N$, we'll compare it to all points $\mathbf{x}_j$ on a shell $\leq N$. Note that technically we don't need to compare to every point on shell $N$, but the analysis is more straightforward this way. For every point $\mathbf{x}_j$ we also need to find the point $\hat{\mathbf{x}}_i \in R^{-1}(\mathbf{x}_i,\delta)$ such that $ \text{gap}^i_\delta$ is minimized.

First, consider $\mathbf{x}_i$ and the set of points $\Theta_i$ that have the same angle as $\mathbf{x}_i$ but are on earlier shells. Then, certainly the point $\mathbf{x}_j$ in $\Theta_i$ that minimizes $\mathbf{\hat{x}}_i\cdot(\mathbf{x}_i - \mathbf{x}_j)$ is the one on shell $N-1$. Furthermore, the worst possible perturbation for that point is exactly $\mathbf{\hat{x}}_i =\mathbf{x}_i$, since $\mathbf{x}_i$ has both coordinates greater than $\mathbf{x}_j$. Therefore, the gap between these points is $\mathbf{x}_i (\mathbf{x}_i - \mathbf{x}_j) = |\mathbf{x}_i|^2 - |\mathbf{x}_i| |\mathbf{x}_j| = \ell_N^2 - \ell_N \ell_{N-1} = \frac{1}{\alpha} \ell_N c_N \in \Omega(c_N)$ (the final relation is because $\ell_N \geq 1/4$ for all $N$). Therefore, when comparing to any $\mathbf{x}_j$ with the same angle as $\mathbf{x}_i$, the $\gap$ is indeed $\Omega(c_N)$. 

Now, consider some $\mathbf{x}_j$ with a different angle. By construction, for any angle $\theta$ such that there exists a point on a shell $\leq N$ with angle $\theta$, there also exists a point $\mathbf{x}_j$ on shell $N$ with angle $\theta$. Therefore, such a $\mathbf{x}_j$ inducing the smallest $\gap$ is clearly on shell $N$. Moreover, it will either be the point immediately clockwise of $\mathbf{x}_i$, or immediately counterclockwise (so as the maximize the dot product and minimize the gap). 

Let  $\theta_i$ be the angle of $\mathbf{x}_i$ on the $N$-th shell, $\theta_j$ be the angle of the $\mathbf{x}_{j}$ we're comparing to (also on the $N$-th shell), and $\hat{\theta}_i$ the angle of a point $\hat{\mathbf{x}}_i \in R^{-1}(\mathbf{x}_i,\delta)$. Our plan is to show first that $\theta_i$ must be far from $\theta_j$ (Claim~\ref{claim:ezpz}), that $\hat{\theta}_i$ must be close to $\theta_{i}$ (Claim~\ref{claim:lb on theta}), and finally that conditioned on these facts, $\mathbf{\hat{x}}_i \cdot(\mathbf{x}_i - \mathbf{x}_{i-1})$ must be large (Claim~\ref{claim:gap in omega}).

\begin{claim}\label{claim:ezpz} $|\theta_i - \theta_j| \geq 3\delta \theta_i$.\end{claim}
\begin{claim}\label{claim:lb on theta} $|\hat{\theta}_i- \theta_i|\leq \delta \theta_i$. \end{claim} 
\begin{claim}\label{claim:gap in omega} Let $\mathbf{x}_i$ and $\mathbf{x}_j$ both lie on shell $N$, and $\mathbf{\hat{x}}_i \in R^{-1}(\mathbf{x_i},\delta)$. Then $\mathbf{\hat{x}}_i\cdot(\mathbf{x}_i - \mathbf{x}_{j-1}) \geq  \frac{2}{\pi^2} \ell^2_N\delta^2 \theta_i^2$. \end{claim}

Given Claim~\ref{claim:gap in omega}, we simply need to observe that most of these terms are in fact independent of $N$ (up to constant factors). Specifically, $\frac{2\delta^2}{\pi^2}$ is obviously independent of $N$. Moreover, $\ell_N \in (1/4, 1)$ for all $N$. So Claim~\ref{claim:gap in omega} indeed concludes that $\mathbf{\hat{x}}_i\cdot(\mathbf{x}_i-\mathbf{x}_j) \in \Omega(\theta_i^2)$. 

This completes the proof of Lemma~\ref{lem:gap}: we have shown that the gap to any point on a lower shell is $\Omega(c_N)$, and the gap to any point on the same shell is $\Omega(\theta_i^2)=\Omega(c_N)$. The final equality is simply because our construction stops placing points on shell $N$ once the angle drops below $\sqrt{c_N}$.\end{proof}


Now that we've computed $\gap_\delta^i$ for any point on shell $N$, we need to compute the number of points on shell $N$ and then we can analyze $\sum_i \gap_\delta^i$. 

\begin{lemma}\label{lem:pointspershell}
The number of points on shell $N$ is $\Theta(\log(1/c_N)) = \Theta(\log (N))$.
\end{lemma}
\begin{proof}
Here, we simply need to figure out the largest value of $k$ such that the angle $\theta_k \geq \sqrt{c_N}$. So we seek the maximum $k$ such that $\pi/2 (1-3\delta)^{k-1} \geq \sqrt{c_N}$. Taking logs of both sides and rearranging yields (note that $\log(1-3\delta) < 0$): $k \leq \frac{\log(c_N)/2 -\ln(\pi/2)}{\ln(1-3\delta)} +1 \in \Theta(\log(1/c_N)) = \Theta(\log (N))$. 
\end{proof}

We can now complete the proof of Lemma~\ref{lemma:points}. We've just shown that our construction has $\Theta(\log N)$ points on shell $N$ (Lemma~\ref{lem:pointspershell}), and each such point has gap at least $\Omega(c_N)$ (Lemma~\ref{lem:gap}). This means that:
 $$\sum_i \gap_\delta^i \in \Omega \left(\sum_N c_N \log N \right) = \Omega \left(\sum_N \frac{1}{N \log N}\right).$$

$\sum_N \frac{1}{N \log N}$ diverges by the integral test, meaning that $\sum_i \gap_\delta^i$ also diverges. 
\end{proof}




Theorem~\ref{thm:LBRect} follows from combining Lemma~\ref{lemma:lb auction} and Lemma~\ref{lemma:points}. A technical takeaway from this section is that the ability to construct bad examples is heavily tied to the ability to pack points in a sphere ``sufficiently far away'' from each other. Without any smoothing, it's possible to pack infinitely many such points~\cite{HartN13}, because the points claim none of the region around them. With rectangle smoothing, the points now claim some of the region around them, and constructions become trickier, but still exist (note, in particular, that because our choice of angles form a geometric progression, our sequence of points is heavily concentrated near the x-axis, where the angle is barely perturbed). On the other hand, in the Square-Shift model, the points now claim a large region around them and it becomes unclear how to pack so many points. Of course, this is just intuition for why these specific constructions no longer work; the following sections show that the Square-Shift model indeed allows for smoothed-finite approximations.


%% file: angle_shift_newer.tex
In this section we present our positive results for the Square-Shift model. We start with a complete proof in a toy ``Angle-Shift'' model. This will help us highlight some of the key insights without the technical barriers. In this model a buyer's value $(x,y) = (r \sin \theta, r \cos \theta)$ is perturbed to a point $( r \sin(\theta+\varepsilon), r \cos(\theta + \varepsilon) )$, where $\varepsilon$ is drawn from $U[-\delta,\delta]$. In other words, we output a vector $(\hat{x}, \hat{y})$ with the same length as $(x,y)$, i.e. $x^2+y^2 = \hat{x}^2+\hat{y}^2$, whose angle is uniformly distributed in the interval $[\max(0,\theta-\delta), \min(\pi/2, \theta+\delta)]$, where $\theta$ is the angle of $(x,y)$ (See Figure~\ref{fig:sub1} in Appendix~\ref{app:table}).

In Section~\ref{subsec:all our results} we restate all our results for the Square-Shift model.
The proof for a single buyer and two items  can be found in Section~\ref{section:missing from one buyer two items}.
We present the proof for multiple buyers and two items in Section~\ref{subsec: square many two}.
We include the single buyer, multi-item case in Section~\ref{sec:one agent many items}, and the multi-buyer, multi-item case (with arbitrary correlation across everything) in Section~\ref{section:generalCase}. Throughout this section we assume familiarity with polar coordinates. For a brief review, see Appendix~\ref{app:all upper bounds}.

\subsection{Angle-Shift Upper Bounds}\label{subsec:angle shift}

In this section we describe our approach for the Angle-Shift model. We show that for this perturbation model, bundling recovers a fraction of the optimal revenue for a single buyer interested in two correlated items. This proof highlights some key insights behind Theorem~\ref{thm: n=1,m=2 box} without most of the technical barriers. 

\begin{restatable}{thm}{angleUB}
\label{thm:angle theorem new}
For the Two-Dimensional Angle-Shift, $\rev{\wiggleD} \leq \frac{\pi}{2\delta} \brev{\wiggleD}.$
\end{restatable}

\begin{proof}
We break the proof into two big steps. The first step will be almost identical for all our positive results. In the second step we use model-specific facts.

\noindent\textbf{Step one.} We begin by writing the revenue of $\wiggleD$ in polar coordinates. Below, $g$ denotes the density of $\realD$ in polar coordinates, and $\hat{g}$ denotes the density of $\wiggleD$ in polar coordinates. Recall that $g$ is \emph{not} the same as $f$. 
We also use the notation $g(\theta)$ to denote $\int_0^\infty g(w, \theta)dw$, and $g(r|\theta)$ to denote  $g(r, \theta)/g(\theta)$ (note that $g(r|\theta)$ is the density of a one-dimensional distribution, which is the continuous analog of conditioning on $\theta$). $p^*(r, \theta)$ denotes the payment of a buyer with values $(r\cos(\theta), r \sin(\theta))$ in the optimal mechanism.
\[ \rev{\wiggleD} = \int_{\theta = 0}^{\frac{\pi}{2}} \int_{r=0}^{\infty} p^*(r,\theta) \hat{g}(r,\theta) dr d\theta = \int_{\theta = 0}^{\frac{\pi}{2}} \hat{g}(\theta) \int_{r=0}^{\infty} p^*(r,\theta) \hat{g}(r|\theta)dr d\theta. \]

Let $\wiggleD_\theta$ be the single parameter distribution from which we first sample $r$ according to the density function $\hat{g}(r|\theta)$, and then output the vector $(r\cos{\theta}, r\sin{\theta})$. Then, notice that the inner most integral is the expected payments according to $p^*$ for $\wiggleD_\theta$. Observe also that if our optimal mechanism is truthful on the entire domain $\mathbb{R}^2$, it is certainly also truthful on the domain restricted to $\theta$. Therefore, this expected payment is upper bounded by $\rev{\wiggleD_\theta}$, the revenue of the optimal truthful mechanism for the same distribution. Furthermore, since this is a single parameter distribution, its revenue is equal to the revenue of a posted price auction (with some reserve $r_\theta$)~\cite{Myerson81}.

\begin{align}
\rev{\wiggleD} &= \int_{\theta = 0}^{\frac{\pi}{2}} \hat{g}(\theta) \int_{r=0}^{\infty} p^*(r,\theta) \hat{g}(r|\theta)dr d\theta \leq \int_{\theta = 0}^{\frac{\pi}{2}} \hat{g}(\theta) \rev{\wiggleD_{\theta}} d\theta \label{eq:ub step}\\
&= \int_{\theta = 0}^{\frac{\pi}{2}} \hat{g}(\theta) \left( r_\theta \int_{r = \frac{r_\theta}{\sin{\theta} + \cos{\theta}}}^{\infty} \hat{g}(r|\theta) dr \right) d\theta = \int_{\theta = 0}^{\frac{\pi}{2}}  r_\theta \int_{r = \frac{r_\theta}{\sin{\theta} + \cos{\theta}}}^{\infty} \hat{g}(r,\theta) dr d\theta \notag.
\end{align}

In going from line one to two above, we are observing that $\rev{\wiggleD_\theta}$ integrates the density of $\wiggleD_\theta$ over all $r$ with $r(\sin \theta + \cos \theta) \geq r_\theta$. The final step is simply rearranging for convenience later. This concludes step one, which is simply to upper bound the optimal revenue by an integral over one-dimensional revenues. Step Two follows, which appeals to the model at hand to unpack $\hat{g}(r,\theta)$.
\noindent\textbf{Step two: model specific analysis.} For the Angle-Shift model we have that

\[
\hat{g}(r,\theta) = \int_{\phi=\max(0,\theta-\delta)}^{\min(\pi/2, \theta+\delta)}  \frac{g(r,\phi)}{\min(\pi/2,\theta+\delta)-\max(0,\theta-\delta)}d\phi \leq \int_{\phi=\max(0,\theta-\delta)}^{\min(\pi/2, \theta+\delta)}	  \frac{g(r,\phi)}{\delta}  d\phi \leq \frac{g(r)}{\delta}.
\]

\noindent Above, the equality is simply because each $(r, \phi)$ when sampled from $g$ sends a $\Theta(1/\delta)$ fraction of its density to $(r,\theta)$. The first inequality is just observing that the fraction of $(r,\phi)$'s density that goes to $(r, \theta)$ is at most $1/\delta$ (it is exactly $1/\delta$ unless $\phi$ is near an axis). The final inequality just observes that the third term is integrating over a proper subset of possible $\phi$, so the integral is certainly upper bounded by integrating the entire region from $0$ to $\pi/2$ (similarly to above, we let $g(r) = \int_0^{\pi/2} g(r, \theta)d\theta$). This concludes step two, which is simply to use the specifics of the model to upper bound the density of the smoothed distribution by that of the original distribution.

We combine steps one and two directly to continue our derivation:
\begin{align*}
\rev{\wiggleD} &\leq \frac{1}{\delta }\int_{\theta = 0}^{\frac{\pi}{2}}  r_\theta \int_{r = \frac{r_\theta}{\sin{\theta} + \cos{\theta}}}^{\infty} g(r) dr d\theta \\
&\leq \frac{1}{\delta }\int_{\theta = 0}^{\frac{\pi}{2}}  r_\theta Pr\left[ \text{a sample $\mathbf{v} \sim \realD$ has $\|\mathbf{v}\|_2 \geq \frac{r_\theta}{\sin{\theta} + \cos{\theta}}$} \right]   d\theta \\
&= \frac{1}{\delta }\int_{\theta = 0}^{\frac{\pi}{2}}  r_\theta Pr\left[ \text{a sample $\mathbf{v} \sim \wiggleD$ has $\|\mathbf{v}\|_2 \geq \frac{r_\theta}{\sin{\theta} + \cos{\theta}}$} \right]   d\theta \\
&= \frac{1}{\delta }\int_{\theta = 0}^{\frac{\pi}{2}}  r_\theta Pr\left[ \text{a sample $\mathbf{v} \sim \wiggleD$ has $\|\mathbf{v}\|_1 \geq r_\theta$} \right]   d\theta \\
&\leq \frac{1}{\delta }\int_{\theta = 0}^{\frac{\pi}{2}}  \brev{\wiggleD}   d\theta = \frac{\pi}{2 \delta } \brev{\wiggleD}.  \qedhere
\end{align*}

\end{proof}

This completes the proof for the Angle-Shift model. Note that above the first inequality simply combines steps one and two. The second observes that integrating the density of a distribution yields (one minus) its CDF. The next line observes that the described event occurs with equal probability for samples from $\realD$ and $\wiggleD$ (this is the step where angle-shifting saves some technical work over square-shifting). The fourth is basic geometry. The final inequality notes that the expression we are integrating over is the revenue of a posted-price mechanism for the bundle (specifically, with reserve $r_\theta$) and in particular can be no better than $\brev{\wiggleD}$ itself. 

Let us now highlight the key step of the proof where we make use of the smoothed model. Step one applies for any distribution (even unsmoothed), and is in general \emph{very} loose. In particular, it could be as high as the full welfare for some distributions. Smoothing gets us mileage in the first half of step two where we transition from terms that depend on $\hat{g}$ to terms that depend on $g$. Mathematically, we capitalize on the following: if there were no angle-shifting, then we would simply have had $\hat{g}(r, \theta) = g(r,\theta)$ instead of the integral with respect to $\phi$. This is important because {integrals over density functions are probabilities}, but {density functions themselves are not probabilities}! That is, if we concentrate on the first part of step two, this is claiming that an integral of a non-negative function over a smaller range is upper bounded by the integral of the same function over a larger range. Without smoothing, we would instead just have a single point of that function, which \emph{can't} generically be upper bounded by an integral over any region.

%% file: state_results.tex
\subsection{Square-Shift: Our Results}\label{subsec:all our results}


We remind the reader of our results for the Square-Shift model below. Section~\ref{subsec:angle shift} captures the intuition for one core step which appears in each of the proofs, but in order to prove our main positive results we still need to overcome a number of technical obstacles. 


\mainUB*


\begin{restatable}{thm}{singleUB}
\label{thm:multibidderUB}

In the Square-Shift model, for $n$ additive buyers and $2$ items:
\begin{align*}
&\textstyle \drev \wiggleD \leq \frac{ 9\sqrt{2}(1+\delta)^3 \log(1+\delta)}{ \delta^2 } \secondprev{\wiggleD} \in O(\secondprev{\wiggleD}/\delta) \text{(if buyers are correlated),}\\
& \textstyle \BICrev \wiggleD \leq \frac{ 9\sqrt{2}(1+\delta)^3 \log(1+\delta)}{ \delta^2 } \secondprev{\wiggleD} \in O(\secondprev{\wiggleD}/\delta) \text{(if buyers are independent)}.
\end{align*}
\end{restatable}


\singlebidUB*

\multiUB*

%% file: box_beyond.tex
\subsection{Square-Shift: One buyer, Two Items.}\label{section:missing from one buyer two items}

\begin{proof}[Proof of Theorem~\ref{thm: n=1,m=2 box}]

Our first few steps are identical to the Angle-Shift model. Recall that the first step on the Angle-Shift model was simply to upper bound the optimal multi-dimensional revenue by that of multiple single-dimensional distributions. We then write the revenue of a single-dimensional distribution as the revenue of a posted-price.

\begin{align}
\rev{ \wiggleD } &= \int_{\theta=0}^{\pi/2} \int_{r=0}^{\infty} p^*(r,\theta) \hat{g}(r,\theta) dr d\theta = \int_{\theta=0}^{\pi/2} \int_{r=0}^{\infty} p^*(r,\theta) \hat{g} (\theta) \hat{g}(r|\theta) dr d\theta \notag\\ 
&\leq^{\text{Lemma~\ref{lemma:anglefocus}}} \int_{\theta=0}^{\pi/2} \hat{g} (\theta)\rev{\wiggleD_{\theta}} d\theta = \int_{\theta=0}^{\pi/2} \hat{g} (\theta) \left( r_{\theta} \int_{r=\frac{r_\theta}{\cos{\theta} + \sin{\theta}}}^{\infty} \hat{g}(r|\theta) dr \right) d\theta \notag\\ 
&= \int_{\theta=0}^{\pi/2} r_\theta \int_{r=\frac{r_\theta}{\cos{\theta} + \sin{\theta}}}^{\infty}  \hat{g}(r, \theta) dr d\theta . \label{eq:continue box}
\end{align}

We now want to proceed with the analog of the second step: upper bound $\int_{r=\frac{r_\theta}{\cos{\theta} + \sin{\theta}}}^{\infty}  \hat{g}(r, \theta) dr$ in terms of the original distribution. We start by moving from polar coordinates to cartesian coordinates, using the transformation $\hat{g}(r, \theta) = r \hat{f}(r \cos{\theta},r \sin{\theta})$. Then, we can connect $\hat{f}$, the density of $\wiggleD$, to $f$, the density of $\realD$, using the following:

\begin{equation}
\label{eq:square shift bound f_hat}
\begin{aligned}
\hat{f}\left(x,y\right) & =  \int_{0}^{\infty} \int_{0}^{\infty} \frac{f(a,b) I\{ (a,b) \in R(x,y) \} da db}{\delta^2 \max\{a,b\}^2} \\
 & \leq \frac{(1+\delta)^2}{\delta^2 \max\{x,y\}^2} \int_{a=0}^{\infty} \int_{b=0}^{\infty} f(a,b) I\{ (a,b) \in R(x,y) \} db da, \\
\end{aligned}
\end{equation}
where $I\{ (a,b) \in R(x,y) \}$ is an indicator for the event that $(a,b)$ belongs in $R(x,y)$. The inequality follows from the fact that for $I\{(a,b) \in R(x,y)\}$ to be 1 we must have that $\max\{a,b\} \geq \frac{\max\{x,y\}}{1+\delta}$. 

\begin{align*}
&\int_{r=\frac{r_\theta}{\cos{\theta} + \sin{\theta}}}^{\infty}  \hat{g}(r, \theta) dr =  \int_{r=\frac{r_\theta}{\cos{\theta} + \sin{\theta}}}^{\infty} r \hat{f}(r \cos \theta, r \sin \theta) dr \\ 
&\quad \leq \frac{(1+\delta)^2}{\delta^2} \int_{r=\frac{r_\theta}{\cos{\theta} + \sin{\theta}}}^{\infty} \frac{r}{r^2 \max\{\cos \theta, \sin \theta\}^2}  \int_{a=0}^{\infty} \int_{b=0}^{\infty} f(a,b) I\{ (a,b) \in R(r,\theta) \} db da dr \\ 
&\quad \leq \frac{2(1+\delta)^2}{\delta^2} \int_{r=\frac{r_\theta}{\cos{\theta} + \sin{\theta}}}^{\infty} \frac{1}{r}  \int_{a=0}^{\infty} \int_{b=0}^{\infty} f(a,b) I\{ (a,b) \in R(r,\theta) \} db da dr .\quad \left[ \max\{ \cos \theta, \sin \theta\}^2 \geq \frac{1}{2} \right]
\end{align*}

This is starting to look familiar to the expression we had at the end of step two for Angle-Shift. The goal of the next few steps is to replace the innermost integral by the integral of the original density on some controlled region and get rid of the indicator function.
At this end of this process, we hope to get a double integral integral (with respect to $a$ and $b$) that we can interpret as a probability.
Our first step is to apply Fubini's Theorem to swap the order of integration:
 
\begin{equation*}
\begin{aligned}
\int_{r=\frac{r_\theta}{\cos{\theta} + \sin{\theta}}}^{\infty}  \hat{g}(r, \theta) dr & \leq \frac{2(1+\delta)^2}{\delta^2}  \int_{a=0}^{\infty} \int_{b=0}^{\infty} f(a,b) \int_{r=\frac{r_\theta}{\cos{\theta} + \sin{\theta}}}^{\infty} \frac{I\{ (a,b) \in R(r,\theta) \}}{r} dr db da. \\
\end{aligned}
\end{equation*}

Now we argue about values of $a$ and $b$ where the indicator $I\{ (a,b) \in R(r,\theta) \}$ is non-zero.

\begin{claim}\label{claim:a b length}
$\forall (a,b) \in R(r,\theta)$, if $r \geq \frac{r_\theta}{\cos{\theta} + \sin{\theta}}$ then $a+b \geq \frac{r_\theta}{(1+\delta)(\cos{\theta} + \sin{\theta})}$.
\end{claim}

\begin{proof}
Recall that $(a,b) \in R(r,\theta)$ means that $(a,b)$ that could map to a point with length $r$ and angle $\theta$. 
The length $\ell_{(a,b)}$ of a point $(a,b)$ that could map to $(r,\theta)$ via our perturbing process must be at least $\frac{r}{1+\delta}$.
Furthermore, this length $\ell_{(a,b)}$ is at most $a+b$. Combining we get $a+b \geq \ell_{(a,b)} \geq \frac{r}{1+\delta} \geq \frac{r_\theta}{(1+\delta)(\cos{\theta} + \sin{\theta})}$.
\end{proof}

We apply Claim~\ref{claim:a b length} to change the limit of integration:

\begin{multline*}
\int_{r=\frac{r_\theta}{\cos{\theta} + \sin{\theta}}}^{\infty}  \hat{g}(r, \theta) dr \leq \\ \frac{2(1+\delta)^2}{\delta^2}  \int_{a = 0}^{\infty} \int_{b = \frac{r_\theta}{(1+\delta)(\cos{\theta} + \sin{\theta})} - a}^{\infty} f(a,b) \int_{r=\frac{r_\theta}{\cos{\theta} + \sin{\theta}}}^{\infty} \frac{I\{ (a,b) \in R(r,\theta) \}}{r} dr db da.
\end{multline*}

We can now upper bound the value of the inner most integral.

\begin{claim}
\label{claim:singledimintbound}
For the current Square-Shift model, for any $(a,b) \in \mathbb{R}^2$,  
$$\int_{r=\frac{r_\theta}{\cos{\theta} + \sin{\theta}}}^{\infty} \frac{1}{r} I\{ (a,b) \in R(r,\theta) \} dr \leq \log(1+\delta).$$ 
\end{claim}

\begin{proof}
For a fixed angle $\theta$ and point $(a,b)$, in order to find out where the indicator is non-zero we need to find the set of points on the line with slope $\theta$ that intersect the set of points to which $(a,b)$ can map to. The line will first intersect the square defined by $(a,b)$ at some point $(x_1,y_1)$ with length $\ell_1$ (if it ever intersects), and leave the square at some point $(x_2,y_2)$ with length $\ell_2$. The ratio between $\ell_1$ and $\ell_2$ is at most $(1+\delta)$: 

$$ \int_{r=\frac{r_\theta}{\cos{\theta} + \sin{\theta}}}^{\infty} \frac{1}{r} I\{ (a,b) \in R(r,\theta) \} dr \leq \int_{r=\ell_1}^{(1+\delta)\ell_1} \frac{1}{r} dr. $$

The claim follows from integrating the simplified upper bound.   
\end{proof}

Once the indicator has been dealt with, the remaining upper bound looks similar to that at the end of step two in the Angle-Shift model. The next couple of lines interpret the double integral as the probability that the buyer drew a sample whose bundle value was above the reserve price of that distribution. This is the moral analog of final step in the Angle-Shift model.
Applying  Claim~\ref{claim:singledimintbound} we have

\begin{align*}
&\int_{r=\frac{r_\theta}{\cos{\theta} + \sin{\theta}}}^{\infty}  \hat{g}(r, \theta) dr \leq \frac{2(1+\delta)^2 \log{(1+\delta)}}{\delta^2}  \int_{a = 0}^{\infty} \int_{b = \frac{r_\theta}{(1+\delta)(\cos{\theta} + \sin{\theta})} - a}^{\infty} f(a,b) db da \\
&\quad= \frac{2(1+\delta)^2 \log{(1+\delta)}}{\delta^2} Pr\left[ \text{sample from $\realD$ has sum of values at least $\frac{r_\theta}{(1+\delta)(\cos{\theta} + \sin{\theta})}$} \right] \\
&\quad \leq \frac{2(1+\delta)^2 \log{(1+\delta)}}{\delta^2} Pr\left[ \text{sample from $\wiggleD$ has sum of values at least $\frac{r_\theta}{(1+\delta)(\cos{\theta} + \sin{\theta})}$} \right],
\end{align*}

where we used the fact that perturbing only increases the value for an item.

The remaining step is to combine steps one and two and interpret the right hand side as a possible bundling mechanism. This can be further upper bounded by the revenue of the optimal bundling scheme. So plugging back in~\ref{eq:continue box}:

\begin{align*}
\rev{ \wiggleD } &\leq \frac{2(1+\delta)^2 \log{(1+\delta)}}{\delta^2}  \int_{\theta=0}^{\pi/2} r_\theta  Pr\left[ \hat{v}_1 + \hat{v}_2 \geq \frac{r_\theta}{(1+\delta)(\cos{\theta} + \sin{\theta})} \right] d\theta \\
&\leq \frac{2(1+\delta)^2 \log{(1+\delta)}}{\delta^2}  \int_{\theta=0}^{\pi/2} r_\theta  Pr\left[ \hat{v}_1 + \hat{v}_2 \geq \frac{r_\theta}{(1+\delta)\sqrt{2}} \right] d\theta \\
&= \frac{2\sqrt{2}(1+\delta)^3 \log{(1+\delta)}}{\delta^2} \int_{\theta=0}^{\pi/2} \frac{r_\theta}{(1+\delta)\sqrt{2}}  Pr\left[ \hat{v}_1 + \hat{v}_2 \geq \frac{r_\theta}{(1+\delta)\sqrt{2}} \right] d\theta.
\end{align*}

\begin{claim}\label{claim:brev bigger than pricing}
For any price $r$, $\brev{\wiggleD} \geq r Pr\left[ \hat{v}_1 + \hat{v}_2 \geq r \right]$.
\end{claim}

\begin{proof}
$\brev{\wiggleD}$ is the revenue maximizing auction among those that sell both items (with allocation probability $1$) whenever the allocation probability is non-zero. $r Pr\left[ \hat{v}_1 + \hat{v}_2 \geq r \right]$ is the revenue of setting a price $r$ for the grand bundle.
\end{proof}

Applying Claim~\ref{claim:brev bigger than pricing} completes the proof of Theorem~\ref{thm: n=1,m=2 box}:

\begin{align*}
\rev{ \wiggleD } &\leq \frac{2 \sqrt{2}(1+\delta)^3 \log{(1+\delta)}}{\delta^2} \int_{\theta=0}^{\pi/2} \brev{\wiggleD} d\theta = \frac{\sqrt{2} \pi (1+\delta)^3 \log{(1+\delta)}}{\delta^2} \brev{\wiggleD}. \qedhere
\end{align*}

\end{proof}

\input{box_new}

\input{general_box}

%% file: box_new.tex
\subsection{Square-Shift: Multiple Buyers, Two Items.}\label{subsec: square many two}

We study the multi-buyer, two-item scenario.
The main difference to the single-buyer case is that interaction between different buyers in optimal auctions (even optimal single-item auctions) can be very complex. Similar to prior work of~\cite{Yao15, CaiDW16}, the main insight (beyond the single-buyer setting) here is to make use of Ronen's auction instead of the optimal auction~\cite{Ronen01}. 

\singleUB*

\begin{proof}

The proof is broken into two steps, similarly to the proofs of Theorem~\ref{thm:angle theorem new} and Theorem~\ref{thm: n=1,m=2 box}. In the first step, we upper bound $\rev{\wiggleD}$ by an expression that involves $\rev{\wiggleD_{\boldsymbol{\theta}}}$, where $\wiggleD_{\boldsymbol{\theta}}$ denotes the $n$-dimensional distribution $\boldsymbol{r}$ conditioned on angles $\boldsymbol{\theta}$ (so below, $\boldsymbol{r}$ denotes the vector of magnitudes for each buyer, and $\boldsymbol{\theta}$ their angles).


\[  \drev{\wiggleD} = \int_{\boldsymbol{\theta}} \hat{g}(\boldsymbol{\theta}) \int_{\boldsymbol{r}} \hat{g}(\boldsymbol{r} | \boldsymbol{\theta})  \sum_{i=1}^n p^*_i\left(\boldsymbol{r}, \boldsymbol{\theta} \right) d\boldsymbol{r} d\boldsymbol{\theta} \leq \int_{\boldsymbol{\theta}} \hat{g}(\boldsymbol{\theta}) \drev{ \wiggleD_{\boldsymbol{\theta}} } d \boldsymbol{\theta},
\]

where in the last inequality we used a generalization of Equation~\ref{eq:ub step}. See Lemma~\ref{lemma:anglefocus} in Appendix~\ref{app:all upper bounds} for a full re-derivation. Again observe that $\wiggleD_{\boldsymbol{\theta}}$ is now a single-parameter instance.


\paragraph{Main intuition.}
From here, we wish to mimic the analysis of the previous section. Unfortunately, optimal multi-buyer mechanisms (even Myerson's~\cite{Myerson81}) don't take a posted-price format, so we use Ronen's~\cite{Ronen01} instead. Fixing angles $\boldsymbol{\theta}$ for all buyers, and lengths $r_{-i}$ for all buyers but $i$, Ronen's mechanism (for the possibly correlated, single parameter distribution $\wiggleD_{\boldsymbol{\theta}}$) makes a take it or leave it offer $r^*_i(\boldsymbol{\theta},\boldsymbol{r}_{-i})$ to buyer $i$, which is larger than the value of any other buyer, i.e. at least $\max_{j \neq i}\left( r_j(\cos \theta_j + \sin \theta_j)\right)$. Ronen proved that $2\secondprev{\wiggleD_{\boldsymbol{\theta}}} \geq\drev{\wiggleD_{\boldsymbol{\theta}}}$; for independent buyers, $2\secondprev{\wiggleD_{\boldsymbol{\theta}}} \geq\BICrev{\wiggleD_{\boldsymbol{\theta}}}$ (and this is the only difference between the correlated and independent buyer proofs). 

Let $w^*_{-i}(\boldsymbol{\theta}_{-i}, \boldsymbol{r}_{-i}) = \max_{j \neq i}\left( r_j(\cos \theta_j + \sin \theta_j)\right)$ be the  highest valuation among all buyers excluding $i$. Then Ronen's mechanism gives buyer $i$ the option to purchase the item at price $r^*_i(\boldsymbol{\theta},\boldsymbol{r}_{-i}) = \argmax_{w \geq w^*_{-i}(\boldsymbol{\theta}_{-i}, \boldsymbol{r}_{-i})} \left( w \cdot \Pr\left[ r_i (\cos \theta_i + \sin \theta_i) \geq w) \right] \right)$. Using this, we get

\begin{align}
&\drev{\wiggleD} \leq 2  \int_{\boldsymbol{\theta}} \hat{g}(\boldsymbol{\theta}) \left( \int_{\boldsymbol{r}} \hat{g}(\mathbf{r}|\boldsymbol{\theta}) \sum_{i=1}^n r^*_i(\boldsymbol{\theta},\boldsymbol{r}_{-i}) I\{ r_i(\cos{\theta_i} + \sin{\theta_i}) \geq  r^*_i(\boldsymbol{\theta},\boldsymbol{r}_{-i}) \} d \mathbf{r} \right) d \boldsymbol{\theta} \notag \\
&= 2 \sum_{i=1}^n \int_{\boldsymbol{\theta}_{-i}, \boldsymbol{r}_{-i}} \hat{g}(\boldsymbol{r}_{-i}, \boldsymbol{\theta}_{-i}) \int_{\theta_i} r^*_i(\boldsymbol{\theta},\boldsymbol{r}_{-i}) \int_{r_i = \frac{r^*_i(\boldsymbol{\theta},\boldsymbol{r}_{-i})}{\cos{\theta_i} + \sin{\theta_i}}}^{\infty} \hat{g}\left( r_i,\theta_i | \boldsymbol{r}_{-i}, \boldsymbol{\theta}_{-i}\right) dr_i d\theta_i d\boldsymbol{r}_{-i} d\boldsymbol{\theta}_{-i}. \label{eq:no good term}
\end{align}

Above, the first line simply replaces $\drev{\wiggleD_{\boldsymbol{\theta}}}$ with a formula for the expected revenue of Ronen's mechanism. The second line performs a few operations: the inner sum is pulled outside the integral, some $\hat{g}$ terms are lumped together for cleanliness, and the integral over an indicator is replaced by an integral over the region where the indicator is one (and the indicator is dropped). 

The brief derivation above captures the main intuition: it's crucial to upper bound the optimal revenue as an integral over posted-price revenues. Of course, as we see next, wrapping up the multi-buyer case has unique subtleties.

\paragraph{The formal argument}


For simplicity, we write $\rev{\wiggleD}$, and don't separate between correlated and independent buyers. Recall though, that for independent buyers $\rev{\wiggleD} = \BICrev{\wiggleD}$ and for correlated buyers, $\rev{\wiggleD} = \drev{\wiggleD}$.

\begin{align}
&\rev{\wiggleD} \leq 2  \int_{\boldsymbol{\theta}} \hat{g}(\boldsymbol{\theta}) \left( \int_{\boldsymbol{r}} \hat{g}(\mathbf{r}|\boldsymbol{\theta}) \sum_{i=1}^n r^*_i(\boldsymbol{\theta},\boldsymbol{r}_{-i}) I\{ r_i(\cos{\theta_i} + \sin{\theta_i}) \geq  r^*_i(\boldsymbol{\theta},\boldsymbol{r}_{-i}) \} d \mathbf{r} \right) d \boldsymbol{\theta} \notag \\
&= 2 \sum_{i=1}^n \int_{\boldsymbol{\theta}_{-i}, \boldsymbol{r}_{-i}} \hat{g}(\boldsymbol{r}_{-i}, \boldsymbol{\theta}_{-i}) \int_{\theta_i} r^*_i(\boldsymbol{\theta},\boldsymbol{r}_{-i}) \int_{r_i = \frac{r^*_i(\boldsymbol{\theta},\boldsymbol{r}_{-i})}{\cos{\theta_i} + \sin{\theta_i}}}^{\infty} \hat{g}\left( r_i,\theta_i | \boldsymbol{r}_{-i}, \boldsymbol{\theta}_{-i}\right) dr_i d\theta_i d\boldsymbol{r}_{-i} d\boldsymbol{\theta}_{-i} \notag \\
\begin{split} \label{eq:first term}
&= 2 \sum_{i=1}^n \int_{\boldsymbol{\theta}_{-i}, \boldsymbol{r}_{-i}} \hat{g}(\boldsymbol{r}_{-i}, \boldsymbol{\theta}_{-i}) \int_{\theta_i} I \{r^*_i(\boldsymbol{\theta},\boldsymbol{r}_{-i}) \leq \sqrt{2}(1+\delta)w^*_i(\boldsymbol{\theta}_{-i},\boldsymbol{r}_{-i}) \}  r^*_i(\boldsymbol{\theta},\boldsymbol{r}_{-i}) \\
&~~~~~~~~~~~~~~~~~~~~~~~~~~~~~~~~ \int_{r_i = \frac{r^*_i(\boldsymbol{\theta},\boldsymbol{r}_{-i})}{\cos{\theta_i} + \sin{\theta_i}}}^{\infty} \hat{g}\left( r_i,\theta_i | \boldsymbol{r}_{-i}, \boldsymbol{\theta}_{-i}\right) dr_i d\theta_i d\boldsymbol{r}_{-i} d\boldsymbol{\theta}_{-i}
\end{split}\\
\begin{split}\label{eq:second term} 
&~+ 2 \sum_{i=1}^n \int_{\boldsymbol{\theta}_{-i}, \boldsymbol{r}_{-i}} \hat{g}(\boldsymbol{r}_{-i}, \boldsymbol{\theta}_{-i}) \int_{\theta_i} I \{r^*_i(\boldsymbol{\theta},\boldsymbol{r}_{-i}) > \sqrt{2}(1+\delta)w^*_i(\boldsymbol{\theta}_{-i},\boldsymbol{r}_{-i}) \}  r^*_i(\boldsymbol{\theta},\boldsymbol{r}_{-i}) \\
&~~~~~~~~~~~~~~~~~~~~~~~~~~~~~~~~ \int_{r_i = \frac{r^*_i(\boldsymbol{\theta},\boldsymbol{r}_{-i})}{\cos{\theta_i} + \sin{\theta_i}}}^{\infty} \hat{g}\left( r_i,\theta_i | \boldsymbol{r}_{-i}, \boldsymbol{\theta}_{-i}\right) dr_i d\theta_i d\boldsymbol{r}_{-i} d\boldsymbol{\theta}_{-i} . 
\end{split}
\end{align}

Recall that $w^*_{-i}(\boldsymbol{\theta}_{-i}, \boldsymbol{r}_{-i}) = \max_{j \neq i}\left( r_j(\cos \theta_j + \sin \theta_j)\right)$ (the highest valuation among all buyers excluding $i$) and $r^*_i(\boldsymbol{\theta},\boldsymbol{r}_{-i}) = \argmax_{w \geq w^*_{-i}(\boldsymbol{\theta}_{-i}, \boldsymbol{r}_{-i})} \left( w \cdot \Pr\left[ r_i (\cos \theta_i + \sin \theta_i) \geq w) \right] \right)$.

\paragraph{Bounding~\ref{eq:first term}.}

We bound~\ref{eq:first term} by the revenue of a second price auction for the grand bundle. Notice that replacing $r^*_i(\boldsymbol{\theta},\boldsymbol{r}_{-i})$ with $w^*_i(\boldsymbol{\theta}_{-i},\boldsymbol{r}_{-i})$ makes the probability of sale increase, and the price decreases by a factor of at most $\frac{1}{\sqrt{2}(1+\delta)}$:

\begin{talign*}
&\ref{eq:first term} \leq 2 \sum_{i=1}^n \int_{\boldsymbol{\theta}_{-i}, \boldsymbol{r}_{-i}} \hat{g}(\boldsymbol{r}_{-i}, \boldsymbol{\theta}_{-i}) \int_{\theta_i} (1+\delta)w^*_i(\boldsymbol{\theta}_{-i},\boldsymbol{r}_{-i}) \int_{r_i = \frac{w^*_i(\boldsymbol{\theta}_{-i},\boldsymbol{r}_{-i})}{\cos{\theta_i} + \sin{\theta_i}}}^{\infty} \hat{g}\left( r_i,\theta_i | \boldsymbol{r}_{-i}, \boldsymbol{\theta}_{-i}\right) dr_i d\theta_i d\boldsymbol{r}_{-i} d\boldsymbol{\theta}_{-i} \\
&=2\sqrt{2}(1+\delta) \sum_{i=1}^n \int_{\boldsymbol{\theta}_{-i}, \boldsymbol{r}_{-i}} \hat{g}(\boldsymbol{r}_{-i}, \boldsymbol{\theta}_{-i}) w^*_i(\boldsymbol{\theta}_{-i},\boldsymbol{r}_{-i}) \int_{\theta_i}  \int_{r_i = \frac{w^*_i(\boldsymbol{\theta}_{-i},\boldsymbol{r}_{-i})}{\cos{\theta_i} + \sin{\theta_i}}}^{\infty} \hat{g}\left( r_i,\theta_i | \boldsymbol{r}_{-i}, \boldsymbol{\theta}_{-i}\right) dr_i d\theta_i d\boldsymbol{r}_{-i} d\boldsymbol{\theta}_{-i} \\
&= 2 \sqrt{2}(1+\delta) \secondpricerev{\wiggleD} \leq 2\sqrt{2} (1+\delta) \secondprev{\wiggleD},
\end{talign*}

where $\secondpricerev{\wiggleD}$ is the optimal revenue that can be attained from using a second price auction for the grand bundle, i.e. the auction where each buyer bids her value for the grand bundle, the buyer with the highest bid wins, and the payment is the second highest bid.

\paragraph{Bounding~\ref{eq:second term}.}

Our approach for bounding~\ref{eq:second term} will be similar to the single buyer case. Specifically, we first upper bound the inner most integral, for fixed $\boldsymbol{\theta}$ and $\boldsymbol{r}_{-i}$, by switching to cartesian coordinates, applying the formula to go from $\hat{f}$ to $f$ for our perturbed model, changing order of integration using Fubini's theorem, and then arguing about the size of the area that gets mapped to the point $(r_i \cos{\theta_i}, r_i \sin{\theta_i})$:

\begin{talign*}
&\int_{r_i = \frac{r^*_i(\boldsymbol{\theta},\boldsymbol{r}_{-i})}{\cos{\theta_i} + \sin{\theta_i}}}^{\infty} \hat{g}\left( r_i,\theta_i | \boldsymbol{r}_{-i}, \boldsymbol{\theta}_{-i}\right) dr_i  = \int_{r_i = \frac{r^*_i(\boldsymbol{\theta},\boldsymbol{r}_{-i})}{\cos{\theta_i} + \sin{\theta_i}}}^{\infty} r_i \hat{f}\left( r_i \cos{\theta_i}, r_i \sin{\theta_i} | \boldsymbol{r}_{-i}, \boldsymbol{\theta}_{-i}\right) dr_i \\
&~~\leq \frac{2(1+\delta)^2}{\delta^2} \int_{r_i = \frac{r^*_i(\boldsymbol{\theta},\boldsymbol{r}_{-i})}{\cos{\theta_i} + \sin{\theta_i}}}^{\infty} \frac{1}{r_i}  \int_{a=0}^{\infty} \int_{b=0}^{\infty} f(a,b | \boldsymbol{r}_{-i}, \boldsymbol{\theta}_{-i}) I\{ (a,b) \in R(r_i,\theta_i) \} db da dr_i \\
& \leq \frac{2(1+\delta)^2}{\delta^2}  \int_{a=0}^{\infty} \int_{b=0}^{\infty} f(a,b | \boldsymbol{r}_{-i}, \boldsymbol{\theta}_{-i}) \int_{r_i = \frac{r^*_i(\boldsymbol{\theta},\boldsymbol{r}_{-i})}{\cos{\theta_i} + \sin{\theta_i}}}^{\infty} \frac{I\{ (a,b) \in R(r_i,\theta_i) \}}{r} dr db da \\
&\leq \frac{2(1+\delta)^2}{\delta^2}  \int_{a = 0}^{\infty} \int_{b = \frac{r^*_i(\boldsymbol{\theta},\boldsymbol{r}_{-i})}{(1+\delta)(\cos{\theta_i} + \sin{\theta_i})} - a}^{\infty} f(a,b | \boldsymbol{r}_{-i}, \boldsymbol{\theta}_{-i}) \int_{r_i = \frac{r^*_i(\boldsymbol{\theta},\boldsymbol{r}_{-i})}{\cos{\theta_i} + \sin{\theta_i}}}^{\infty} \frac{I\{ (a,b) \in R(r_i,\theta_i) \}}{r} dr_i db da \\
&\leq \frac{2(1+\delta)^2 \log{(1+\delta)}}{\delta^2}  \int_{a = 0}^{\infty} \int_{b = \frac{r^*_i(\boldsymbol{\theta},\boldsymbol{r}_{-i})}{(1+\delta)(\cos{\theta_i} + \sin{\theta_i})} - a}^{\infty} f(a,b | \boldsymbol{r}_{-i}, \boldsymbol{\theta}_{-i}) db da \\
&= \frac{2(1+\delta)^2 \log{(1+\delta)}}{\delta^2} Pr\left[ \text{$x \sim \realD$ has $\|x\|_1 \geq \frac{r^*_i(\boldsymbol{\theta},\boldsymbol{r}_{-i})}{(1+\delta)(\cos{\theta_i} + \sin{\theta_i})}$} \mid \boldsymbol{r}_{-i}, \boldsymbol{\theta}_{-i} \right] \\
&\leq \frac{2(1+\delta)^2 \log{(1+\delta)}}{\delta^2} Pr\left[ \text{$x \sim \wiggleD$ has $\|x\|_1 \geq \frac{r^*_i(\boldsymbol{\theta},\boldsymbol{r}_{-i})}{(1+\delta)\sqrt{2}}$} \mid \boldsymbol{r}_{-i}, \boldsymbol{\theta}_{-i} \right].
\end{talign*}

%

We can now bound~\ref{eq:second term}:
\begin{talign*}
\ref{eq:second term} &\leq \frac{4(1+\delta)^2 \log{(1+\delta)}}{\delta^2} \sum_{i=1}^n \int_{\boldsymbol{\theta}_{-i}, \boldsymbol{r}_{-i}} \hat{g}(\boldsymbol{r}_{-i}, \boldsymbol{\theta}_{-i}) \int_{\theta_i} I \{\frac{r^*_i(\boldsymbol{\theta},\boldsymbol{r}_{-i})}{(1+\delta)\sqrt{2}} > w^*_i(\boldsymbol{\theta}_{-i},\boldsymbol{r}_{-i}) \}   \\
&~~~~~~~~~~~~~~~~~~~~~~~~~~~~~~~~~~~~~~~~~~ r^*_i(\boldsymbol{\theta},\boldsymbol{r}_{-i})  Pr\left[ \text{$x_i \sim \wiggleD_i$ has $\|x_i\|_1 \geq \frac{r^*_i(\boldsymbol{\theta},\boldsymbol{r}_{-i})}{(1+\delta)\sqrt{2}}$} \mid \boldsymbol{r}_{-i}, \boldsymbol{\theta}_{-i} \right] d\theta_i d\boldsymbol{r}_{-i} d\boldsymbol{\theta}_{-i} \\
&= \frac{4\sqrt{2}(1+\delta)^3 \log{(1+\delta)}}{\delta^2} \sum_{i=1}^n \int_{\boldsymbol{\theta}_{-i}, \boldsymbol{r}_{-i}} \hat{g}(\boldsymbol{r}_{-i}, \boldsymbol{\theta}_{-i}) \int_{\theta_i} I \{\frac{r^*_i(\boldsymbol{\theta},\boldsymbol{r}_{-i})}{(1+\delta)\sqrt{2}} > w^*_i(\boldsymbol{\theta}_{-i},\boldsymbol{r}_{-i}) \}   \\
&~~~~~~~~~~~~~~~~~~~~~~~~~~~~~~~~~~~~~~~~~~ \frac{r^*_i(\boldsymbol{\theta},\boldsymbol{r}_{-i})}{(1+\delta)\sqrt{2}}  Pr\left[ \text{$x_i \sim \wiggleD_i$ has $\|x_i\|_1 \geq \frac{r^*_i(\boldsymbol{\theta},\boldsymbol{r}_{-i})}{(1+\delta)\sqrt{2}}$} \mid \boldsymbol{r}_{-i}, \boldsymbol{\theta}_{-i} \right] d\theta_i d\boldsymbol{r}_{-i} d\boldsymbol{\theta}_{-i} \\
&\leq \frac{2\sqrt{2} \pi (1+\delta)^3 \log{(1+\delta)}}{\delta^2} \sum_{i=1}^n \int_{\boldsymbol{\theta}_{-i}, \boldsymbol{r}_{-i}} \hat{g}(\boldsymbol{r}_{-i}, \boldsymbol{\theta}_{-i}) I \{\frac{r^*_i(\boldsymbol{\theta}_{-i},\boldsymbol{r}_{-i})}{(1+\delta)\sqrt{2}} > w^*_i(\boldsymbol{\theta}_{-i},\boldsymbol{r}_{-i}) \}   \\
&~~~~~~~~~~~~~~~~~~~~~~~~~~~~~~~~~~~~~~~~~~ \frac{r^*_i(\boldsymbol{\theta}_{-i},\boldsymbol{r}_{-i})}{(1+\delta)\sqrt{2}}  Pr\left[ \text{$x_i \sim \wiggleD_i$ has $\|x_i\|_1 \geq \frac{r^*_i(\boldsymbol{\theta}_{-i},\boldsymbol{r}_{-i})}{(1+\delta)\sqrt{2}}$} \mid \boldsymbol{r}_{-i}, \boldsymbol{\theta}_{-i} \right] d\boldsymbol{r}_{-i} d\boldsymbol{\theta}_{-i},
\end{talign*}

where the last inequality follows from switching from $r^*_i(\boldsymbol{\theta},\boldsymbol{r}_{-i})$ to $r^*_i(\boldsymbol{\theta}_{-i},\boldsymbol{r}_{-i})$ and paying an additional factor of $\pi/2$ (to remove the integral with respect to $\theta_i$). 

Let $\mathcal{A}$ be the set of auctions that sells the grand bundle as a whole, and only sells to the highest bidder.
The auction in the RHS offers the grand bundle to buyer $i$ for some price $q_i$, such that $q_i \geq w^*_i(\boldsymbol{\theta}_{-i},\boldsymbol{r}_{-i})$ (and therefore it sells to at most one buyer), and is an auction in $\mathcal{A}$. $\secondprev{\wiggleD}$ is the revenue maximizing auction in $\mathcal{A}$, therefore:

\[ \ref{eq:second term} \leq \frac{2\sqrt{2} \pi (1+\delta)^3 \log{(1+\delta)}}{\delta^2} \secondprev{\wiggleD}. \]

\paragraph{Putting it all together.}
Combining our bound for~\ref{eq:first term} and~\ref{eq:second term} we get:

\begin{align*}
\rev{\wiggleD} &\leq \left( 2 \sqrt{2}(1+\delta)+ \frac{2\sqrt{2}\pi (1+\delta)^3 \log{(1+\delta)}}{\delta^2} \right)\secondprev{\wiggleD} \\
&= \frac{ 2\sqrt{2}(1+\delta) }{ \delta^2 } \left( \delta^2 + \pi (1+\delta)^2 \log(1+\delta) \right) \secondprev{\wiggleD} \\
&\leq \frac{ 9\sqrt{2}(1+\delta)^3 \log(1+\delta)}{ \delta^2 } \secondprev{\wiggleD}.\qedhere
\end{align*}

\end{proof}

%% file: general_box.tex
\subsection{Square-Shift: One buyer, Many Items. }
\label{sec:one agent many items}

\singlebidUB*

\begin{proof}

Our first few steps are the same. In order to write revenue in polar coordinates we need $m-1$ angles $\boldsymbol{\theta} = (\theta_1, \dots, \theta_{m-1})$ and a length $r$:

\begin{align}
&\rev{ \wiggleD } = \int_{\boldsymbol{\theta} \in [0,\pi/2]^{m-1} } \int_{r=0}^{\infty} p^*(r,\boldsymbol{\theta}) \hat{g}(r,\boldsymbol{\theta}) dr d\boldsymbol{\theta}  = \int_{\boldsymbol{\theta} \in [0,\pi/2]^{m-1} } \int_{r=0}^{\infty} p^*(r,\boldsymbol{\theta}) \hat{g} (\boldsymbol{\theta}) \hat{g}(r|\boldsymbol{\theta}) dr d\boldsymbol{\theta} \notag\\ 
&\quad\leq^{\text{Lemma~\ref{lemma:anglefocus}}} \int_{\boldsymbol{\theta} \in [0,\pi/2]^{m-1} } \hat{g} (\boldsymbol{\theta})\rev{\wiggleD_{\boldsymbol{\theta}}} d\boldsymbol{\theta} = \int_{\boldsymbol{\theta} \in [0,\pi/2]^{m-1} } \hat{g} (\boldsymbol{\theta}) \left( r_{\boldsymbol{\theta}} \int_{r= \frac{r_{\boldsymbol{\theta}}}{\sum_{j=1}^{m} \trig_j }}^{\infty} \hat{g}( r | \boldsymbol{\theta}) dr \right) d\boldsymbol{\theta} \notag\\ 
&\quad= \int_{ \boldsymbol{\theta} \in [0,\pi/2]^{m-1} } r_{\boldsymbol{\theta}} \int_{ r=\frac{r_{\boldsymbol{\theta}}}{\sum_{j=1}^{m} \trig_j} }^{\infty}  \hat{g}(r, \boldsymbol{\theta}) dr d\boldsymbol{\theta} . \label{eq:continue general box}
\end{align}

Next, we bound the integral $\int_{ r=\frac{r_{\boldsymbol{\theta}}}{\sum_{j=1}^{m} \trig_j} }^{\infty}  \hat{g}(r, \boldsymbol{\theta}) dr$. Our first step is to go from polar to cartesian coordinates. This time, the transformation is a bit more complex. 
More specifically, given a point $(w,\phi_1,\dots,\phi_{m-1})$ in $m$ dimensions expressed in polar coordinates, the corresponding point in Cartesian coordinates is $\mathbf{x}$, where $x_i = w \cos{\phi_i} \prod_{j=1}^{i-1} \sin{\phi_j}$. The transformation from $\hat{g}$ to $\hat{f}$ is given in Claim~\ref{claim:jacobian}.
Furthermore, we need to upper bound $\hat{f}$ by $f$. For the square model the two are connected by the following equations:

\begin{align}
\hat{f}\left(\mathbf{ \hat{x} } \right) &=  
\int_{\mathbf{x} \in [0,\infty]^m} \frac{f(\mathbf{x}) I\{ \mathbf{x} \in R(\mathbf{\hat{x}}) \} }{\left( \delta \max_{j\in[1,m]} x_j \right)^m} d \mathbf{x} \leq \frac{(1+\delta)^m}{\left( \delta \max_{j\in[1,m]} \hat{x}_j \right)^m} \int_{\mathbf{x} \in [0,\infty]^m} f(\mathbf{x}) I\{ \mathbf{x} \in R(\mathbf{\hat{x}}) \}  d \mathbf{x}, \label{square shift bound f_hat general m}
\end{align}
where $I\{ \mathbf{x} \in R(\mathbf{\hat{x}}) \} $ is an indicator for the event that $\mathbf{x}$ belongs in $R(\mathbf{\hat{x}})$. The inequality follows from the fact that for $I\{ \mathbf{x} \in R(\mathbf{\hat{x}}) \} $ to be 1 we must have that $\max_{j\in[1,m]} x_j \geq \frac{\max_{j\in[1,m]} \hat{x}_j}{1+\delta}$. 

\begin{align*}
&\int_{ r=\frac{r_{\boldsymbol{\theta}}}{\sum_{j=1}^{m} \trig_j} }^{\infty}  \hat{g}(r, \boldsymbol{\theta}) dr = \int_{ r=\frac{r_{\boldsymbol{\theta}}}{\sum_{j=1}^{m} \trig_j} }^{\infty}  r^{m-1} \prod_{j=1}^{m-2} (\sin{\theta_j})^{m-j-1} \hat{f} \left( r \cdot \mathbf{\trig} \right) dr  \\
&\quad \leq \int_{ r=\frac{r_{\boldsymbol{\theta}}}{\sum_{j=1}^{m} \trig_j} }^{\infty}  r^{m-1} \prod_{j=1}^{m-2} (\sin{\theta_j})^{m-j-1} \cdot \\
&~~~~~~~~~~~~~~~~~~~~~~~~~~~~~~~~  \left(  \frac{(1+\delta)^m}{\left( \delta r \max_{j\in[1,m]} \trig_j \right)^m} \int_{\mathbf{x} \in [0,\infty]^m} f\left( \mathbf{x} \right) I\{ \mathbf{x} \in R(r \cdot \mathbf{\trig}) \}  d \mathbf{x}  \right) dr \\
&\quad \leq^{\text{Cl.~\ref{claim:trigfact}}} \left( \frac{(1+\delta) \sqrt{m} }{\delta} \right)^m \prod_{j=1}^{m-2} (\sin{\theta_j})^{m-j-1} \int_{ r=\frac{r_{\boldsymbol{\theta}}}{\sum_{j=1}^{m} \trig_j} }^{\infty} \int_{\mathbf{x} \in [0,\infty]^m} \frac{1}{r} f\left( \mathbf{x} \right) I\{ \mathbf{x} \in R\left( r \cdot \mathbf{\trig} \right) \}  d \mathbf{x}  dr \\
&\quad = \left( \frac{(1+\delta) \sqrt{m} }{\delta} \right)^m \prod_{j=1}^{m-2}(\sin{\theta_j})^{m-j-1}  \int_{\mathbf{x} \in [0,\infty]^m}  f\left( \mathbf{x} \right) \int_{ r=\frac{r_{\boldsymbol{\theta}}}{\sum_{j=1}^{m} \trig_j} }^{\infty}  \frac{1}{r}  I\{ \mathbf{x} \in R\left( r \cdot \mathbf{\trig} \right) \}  dr d \mathbf{x}.
\end{align*}

Next, we argue about values of $\mathbf{x}$ such that the indicator $I\{ \mathbf{x} \in R\left( r \cdot \mathbf{\trig} \right) \}$ is non-zero.

\begin{claim}\label{claim:x length general box}
$\forall \mathbf{x} \in R\left( r \cdot \mathbf{\trig} \right)$, if $r \geq \frac{r_{\boldsymbol{\theta}}}{\sum_{j=1}^{m} \trig_j}$ then $\sum_{j=1}^m x_j \geq \frac{r_\theta}{(1+\delta)(\sum_{j=1}^{m} \trig_j)}$.
\end{claim}

\begin{proof}
Similarly to Claim~\ref{claim:a b length}, note that for any vector $\mathbf{x} \in \mathbb{R}^+_m$, $\sum_{i=1}^m x_i \geq r_{\mathbf{x}} = \sqrt{\sum x_i^2}$. Moreover, any vector $\mathbf{x}$ that can map to $(r, \boldsymbol{\theta})$ must have its own length, $r_{\mathbf{x}}$ be at least $\frac{1}{1+\delta} r$. Combining these two claims with that given on the statement completes the proof. 
\end{proof}

We apply Claim~\ref{claim:x length general box} to change the limit of integration:

\begin{multline*}
\int_{ r=\frac{r_{\boldsymbol{\theta}}}{\sum_{j=1}^{m} \trig_j} }^{\infty}  \hat{g}(r, \boldsymbol{\theta}) dr \leq \left( \frac{(1+\delta) \sqrt{m} }{\delta} \right)^m \prod_{j=1}^{m-2}(\sin{\theta_j})^{m-j-1}  \cdot \\
\cdot \int_{\mathbf{x}: \| \mathbf{x} \|_1 \geq \frac{r_\theta}{(1+\delta)(\sum_{j=1}^{m} \trig_j)} }  f\left( \mathbf{x} \right) \int_{ r=\frac{r_{\boldsymbol{\theta}}}{\sum_{j=1}^{m} \trig_j} }^{\infty}  \frac{1}{r}  I\{ \mathbf{x} \in R\left( r \cdot \mathbf{\trig} \right) \}  dr d \mathbf{x}.
\end{multline*}

We can now upper bound the value of the inner most integral.

\begin{claim}
\label{claim:general box log upper bound}
For the current perturbation model, for any $\forall \mathbf{x} \in R\left( r \cdot \mathbf{\trig} \right)$,  
$$\int_{ r=\frac{r_{\boldsymbol{\theta}}}{\sum_{j=1}^{m} \trig_j} }^{\infty}  \frac{1}{r}  I\{ \mathbf{x} \in R\left( r \cdot \mathbf{\trig} \right) \}  dr \leq \log(1+\delta).$$ 
\end{claim}

\begin{proof} 
The points that $\mathbf{x}$ can map to form a hypercube, and the indicator $I\{ \mathbf{x} \in R\left( r \cdot \mathbf{\trig} \right) \}$ is non-zero for a length $r$, when the point $(r,\boldsymbol{\theta})$ intersects that hypercube. Let $r_{min}$ be the smallest length among all points $(r,\boldsymbol{\theta})$ in the hypercube defined by $\mathbf{x}$. Then, the maximum length cannot be larger than $r_{min}\dot (1+\delta)$. Therefore  
$$\int_{ r=\frac{r_{\boldsymbol{\theta}}}{\sum_{j=1}^{m} \trig_j} }^{\infty}  \frac{1}{r}  I\{ \mathbf{x} \in R\left( r \cdot \mathbf{\trig} \right) \}  dr \leq \int_{ r=r_{min} }^{(1+\delta)r_{min}}  \frac{1}{r}   dr.$$ 
Integrating the latter integral proves the Claim.
\end{proof}

We apply Claim~\ref{claim:general box log upper bound}  to continue our derivation, 

\begin{align*}
&\int_{ r=\frac{r_{\boldsymbol{\theta}}}{\sum_{j=1}^{m} \trig_j} }^{\infty}  \hat{g}(r, \boldsymbol{\theta}) dr \leq \left( \frac{(1+\delta) \sqrt{m} }{\delta} \right)^m \log(1+\delta) \prod_{j=1}^{m-2}(\sin{\theta_j})^{m-j-1} \int_{\substack{ \mathbf{x} \in [0,\infty)^m \text{ such that } \\ \| \mathbf{x} \|_1 \geq \frac{r_{\boldsymbol{\theta}}}{(1+\delta)(\sum_{j=1}^{m} \trig_j)} }}  f\left( \mathbf{x} \right)  d \mathbf{x} \\
&= \left( \frac{(1+\delta) \sqrt{m} }{\delta} \right)^m \log(1+\delta) \prod_{j=1}^{m-2}(\sin{\theta_j})^{m-j-1} 
Pr\left[\text{$\mathbf{x} \sim \realD$ has $\|\mathbf{x}\|_1 \geq \frac{r_{\boldsymbol{\theta}}}{(1+\delta)(\sum_{j=1}^{m} \trig_j)}$ }\right] \\
&\leq \left( \frac{(1+\delta) \sqrt{m} }{\delta} \right)^m \log(1+\delta) \prod_{j=1}^{m-2}(\sin{\theta_j})^{m-j-1} 
Pr\left[\text{$\mathbf{x} \sim \wiggleD$ has $\|\mathbf{x}\|_1 \geq \frac{r_{\boldsymbol{\theta}}}{(1+\delta)(\sum_{j=1}^{m} \trig_j)}$ }\right]\\
&\leq \left( \frac{(1+\delta) \sqrt{m} }{\delta} \right)^m \log(1+\delta) \prod_{j=1}^{m-2}(\sin{\theta_j})^{m-j-1} 
Pr\left[\text{$\mathbf{x} \sim \wiggleD$ has $\|\mathbf{x}\|_1 \geq \frac{r_{\boldsymbol{\theta}}}{(1+\delta)\sqrt{m}}$ }\right],
\end{align*}

where we used the facts that perturbing only increases values, and that $\sum_{j=1}^m \trig_j \leq \sqrt{m}$. Plugging back in Equation~\ref{eq:continue general box}:

\begin{align*}
&\rev{ \wiggleD } \leq \int_{ \boldsymbol{\theta} \in [0,\pi/2]^{m-1} } r_{\boldsymbol{\theta}} \cdot\\
&\left( \left( \frac{(1+\delta) \sqrt{m} }{\delta} \right)^m \log(1+\delta) \prod_{j=1}^{m-2}(\sin{\theta_j})^{m-j-1} 
Pr\left[\text{sample from $\wiggleD$ has $\|\mathbf{x}\|_1 \geq \frac{r_{\boldsymbol{\theta}}}{(1+\delta)\sqrt{m}}$ }\right] \right) d\boldsymbol{\theta} \\
&\leq \left( \frac{(1+\delta) \sqrt{m} }{\delta} \right)^m \log(1+\delta) (1+\delta) \sqrt{m} \int_{ \boldsymbol{\theta} \in [0,\pi/2]^{m-1} } \prod_{j=1}^{m-2}(\sin{\theta_j})^{m-j-1} \brev{\wiggleD} d\boldsymbol{\theta} \\
&=\brev{\wiggleD} \left( \frac{(1+\delta) \sqrt{m} }{\delta} \right)^m (1+\delta) \log(1+\delta)  \sqrt{m} \int_{ \boldsymbol{\theta} \in [0,\pi/2]^{m-1} } \prod_{j=1}^{m-2}(\sin{\theta_j})^{m-j-1}  d\boldsymbol{\theta},
\end{align*}

where we used the fact that $\brev{\wiggleD}$ is at least $r \cdot Pr\left[ \sum_{j=1}^m x_j \geq r \right]$ for all $r$, where $(x_1,\dots,x_m)$ is a sample from $\wiggleD$. The next Claim is used to simplify the integral over the product of sine functions:

\begin{claim}\label{claim:integral of sin}
\[ \int_{ \boldsymbol{\theta} \in [0,\pi/2]^{m-1} } \prod_{j=1}^{m-2}(\sin{\theta_j})^{m-j-1}  d\boldsymbol{\theta} \leq \left( \sqrt{\frac{\pi e}{2}} \right)^{m} \frac{m}{(\sqrt{m})^{m}}. \]
\end{claim}

\begin{proof}
We repeatedly use the fact that $\int_{\theta=0}^{\pi/2} (\sin{\theta})^k d\theta = \frac{\sqrt{\pi}}{2} \frac{\Gamma\left( \frac{k+1}{2}\right)}{\Gamma\left( \frac{k}{2} + 1 \right)}$.

\begin{align*}
&\int_{\theta_1=0}^{\pi/2} (\sin{\theta_1})^{m-2} \int_{\theta_2=0}^{\pi/2} (\sin{\theta_2})^{m-3} \int_{\theta_3=0}^{\pi/2} \dots \int_{\theta_{m-2}=0}^{\pi/2} \sin{\theta_{m-2}} \left( \int_{\theta_{m-1}=0}^{\pi/2} 1 d\theta_{m-1} \right) d\theta_{m-2} \dots d\theta_1 \\
&= \frac{\pi}{2} \int_{\theta_1=0}^{\pi/2} (\sin{\theta_1})^{m-2} \int_{\theta_2=0}^{\pi/2} (\sin{\theta_2})^{m-3} \int_{\theta_3=0}^{\pi/2} \dots \relax \left( \int_{\theta_{m-2}=0}^{\pi/2} \sin{\theta_{m-2}} d\theta_{m-2} \right)  \dots d\theta_1 \\
&= \frac{\pi}{2} \frac{\sqrt{\pi}}{2} \frac{\Gamma\left( 1 \right)}{\Gamma\left( \frac{1}{2} + 1 \right)} \int_{\theta_1=0}^{\pi/2} (\sin{\theta_1})^{m-2} \int_{\theta_2=0}^{\pi/2} (\sin{\theta_2})^{m-3} \int_{\theta_3=0}^{\pi/2} \dots \relax \left( \int_{\theta_{m-3}=0}^{\pi/2} (\sin{\theta_{m-3}})^2 d\theta_{m-3} \right)  \dots d\theta_1 \\
&= \frac{\pi}{2} \left( \frac{\sqrt{\pi}}{2} \right)^2 \frac{\Gamma\left( 1 \right)}{\Gamma\left( \frac{1}{2} + 1 \right)} \frac{\Gamma\left( \frac{1}{2} + 1 \right)}{\Gamma\left( 2 \right)} \int_{\theta_1=0}^{\pi/2} (\sin{\theta_1})^{m-2} \int_{\theta_2=0}^{\pi/2} \dots \relax \left( \int_{\theta_{m-4}=0}^{\pi/2} (\sin{\theta_{m-4}})^3 d\theta_{m-4} \right)  \dots d\theta_1 \\
&= \frac{\pi}{2} \left( \frac{\sqrt{\pi}}{2} \right)^{m-2} \frac{\Gamma\left( 1 \right)}{\Gamma\left( \frac{m-2}{2} + 1 \right)} \\
&= \frac{\pi}{2} \left( \frac{\sqrt{\pi}}{2} \right)^{m-2} \frac{1}{\Gamma\left( \frac{m}{2} \right)}.
\end{align*}

Recall that for integer $n$, $\Gamma(n+1) = n!$, therefore $\Gamma\left( \frac{m}{2} \right) = \left( \frac{m}{2} -1 \right)!$. Using the fact that $k! \geq \left( \frac{k}{e} \right)^{k}$ we get:

\begin{align*}
\frac{\pi}{2} \left( \frac{\sqrt{\pi}}{2} \right)^{m-2} \frac{1}{\Gamma\left( \frac{m}{2} \right)} &= \frac{\pi}{2} \left( \frac{\sqrt{\pi}}{2} \right)^{m-2} \frac{1}{\left( \frac{m}{2} - 1 \right)!}\\
&= \frac{\pi}{2} \left( \frac{\sqrt{\pi}}{2} \right)^{m-2} \frac{\frac{m}{2}}{\left( \frac{m}{2} \right)!}\\
&\leq \frac{\pi}{4} \left( \frac{\sqrt{\pi}}{2} \right)^{m-2} \frac{m}{\left( \frac{m}{2e} \right)^{m/2}}\\
&= \left( \frac{\sqrt{\pi}}{2} \right)^{m} (2e)^{m/2} \frac{m}{\left( \sqrt{m} \right)^m}\\
&= \left( \sqrt{\frac{\pi e}{2}} \right)^{m} \frac{m}{(\sqrt{m})^{m}},
\end{align*}
which concludes the proof of the Claim.
\end{proof}

Applying Claim~\ref{claim:integral of sin} completes the proof of Theorem~\ref{thm: n=1,general m box}:
\begin{align*}
\rev{ \wiggleD } &\leq \brev{\wiggleD} \left( \frac{(1+\delta) \sqrt{m} }{\delta} \right)^m (1+\delta) \log(1+\delta)  \sqrt{m} \left( \sqrt{\frac{\pi e}{2}} \right)^{m} \frac{m}{(\sqrt{m})^{m}} \\
&= \brev{\wiggleD} \left( \sqrt{\frac{\pi e}{2}} \frac{(1+\delta) }{\delta} \right)^m (1+\delta) \log(1+\delta)  m \sqrt{m}. \qedhere
\end{align*}

\end{proof}

%% file: general_case.tex
\subsection{Square-Shift: Many buyers, Many Items. }\label{section:generalCase}

In this section we prove upper bounds for the case of multiple (possibly correlated) buyers and multiple items, in the Square-Shift model.

\multiUB*

\begin{proof}

We write $\boldsymbol{\theta}$ for the vector $(\boldsymbol{\theta}_1,\dots,\boldsymbol{\theta}_n)$, where $\boldsymbol{\theta}_i = (\theta_{i,1},\dots,\theta_{i,m-1})$ is the vector of angles for buyer $i$:

 \[
\rev{\wiggleD} = \int_{\boldsymbol{\theta}} \hat{g}(\boldsymbol{\theta}) \int_{\boldsymbol{r}} \hat{g}(\boldsymbol{r} | \boldsymbol{\theta})  \sum_{i=1}^n p^*_i\left(\boldsymbol{r}, \boldsymbol{\theta} \right) d\boldsymbol{r} d\boldsymbol{\theta} \leq \int_{\boldsymbol{\theta}} \hat{g}(\boldsymbol{\theta}) \rev{ \wiggleD_{\boldsymbol{\theta}} } d \boldsymbol{\theta},
\]

where $\wiggleD_{\boldsymbol{\theta}}$ is a correlated, single parameter distribution where we draw a value vector $\boldsymbol{r} = (r_1,\dots, r_n)$ according to the density function $\hat{g}(\boldsymbol{r} | r_{-i} \boldsymbol{\theta})$, and give each buyer $i$ value $r_i \cdot \trig_k(\boldsymbol{\theta}_i)$ for item $k$.
Ronen's mechanism \cite{Ronen01} for this distribution makes a take it or leave it offer $r^*_i(\boldsymbol{\theta},\boldsymbol{r}_{-i})$ to buyer $i$ for the grand bundle; $r^*_i(\boldsymbol{\theta},\boldsymbol{r}_{-i})$ is larger than the value of any other buyer, i.e. at least $\max_{j \neq i}\left( r_j \sum_{k=1}^m \trig_k(\boldsymbol{\theta}_j)\right)$. This mechanism is a $2$ approximation to the optimal mechanism for $\wiggleD_{\boldsymbol{\theta}}$. Let $w^*_{-i}(\boldsymbol{\theta}_{-i}, \boldsymbol{r}_{-i}) = \max_{j \neq i}\left( r_j\sum_{k=1}^m \trig_k(\boldsymbol{\theta}_j)\right)$ be the second highest valuation among all buyers excluding $i$, and $r^*_i(\boldsymbol{\theta},\boldsymbol{r}_{-i}) = \argmax_{w \geq w^*_{-i}(\boldsymbol{\theta}_{-i}, \boldsymbol{r}_{-i})} \left( w \cdot \Pr\left[ r_i \sum_{k=1}^m \trig_k(\boldsymbol{\theta}_i) \geq w) \right] \right)$ be the price offered by Ronen's mechanism.

\begin{talign}
&\rev{\wiggleD} \leq 2  \int_{\boldsymbol{\theta}} \hat{g}(\boldsymbol{\theta}) \left( \int_{\boldsymbol{r}} \hat{g}(\mathbf{r}|\boldsymbol{\theta}) \sum_{i=1}^n r^*_i(\boldsymbol{\theta},\boldsymbol{r}_{-i}) I\{ r_i\sum_{k=1}^m \trig_k(\boldsymbol{\theta}_i) \geq  r^*_i(\boldsymbol{\theta},\boldsymbol{r}_{-i}) \} d \mathbf{r} \right) d \boldsymbol{\theta} \notag \\
&= 2 \sum_{i=1}^n \int_{\boldsymbol{\theta}_{-i}, \boldsymbol{r}_{-i}} \hat{g}(\boldsymbol{r}_{-i}, \boldsymbol{\theta}_{-i}) \int_{\boldsymbol{\theta}_i} r^*_i(\boldsymbol{\theta},\boldsymbol{r}_{-i}) \int_{r_i = \frac{r^*_i(\boldsymbol{\theta},\boldsymbol{r}_{-i})}{ \sum_{k=1}^m \trig_k(\boldsymbol{\theta}_i) }}^{\infty} \hat{g}\left( r_i, \boldsymbol{\theta}_i | \boldsymbol{r}_{-i}, \boldsymbol{\theta}_{-i}\right) dr_i d\boldsymbol{\theta}^i d\boldsymbol{r}_{-i} d\boldsymbol{\theta}_{-i} \label{eq:no good term multi} \\
\begin{split} \label{eq:first term multi}
&= 2 \sum_{i=1}^n \int_{\boldsymbol{\theta}_{-i}, \boldsymbol{r}_{-i}} \hat{g}(\boldsymbol{r}_{-i}, \boldsymbol{\theta}_{-i}) \int_{\boldsymbol{\theta}_i} I \{r^*_i(\boldsymbol{\theta},\boldsymbol{r}_{-i}) \leq \sqrt{m}(1+\delta)w^*_i(\boldsymbol{\theta}_{-i},\boldsymbol{r}_{-i}) \}  r^*_i(\boldsymbol{\theta},\boldsymbol{r}_{-i}) \\
&~~~~~~~~~~~~~~~~~~~~~~~~~~~~~~~~ \int_{r_i = \frac{r^*_i(\boldsymbol{\theta},\boldsymbol{r}_{-i})}{\sum_{k=1}^m \trig_k(\boldsymbol{\theta}_i)}}^{\infty} \hat{g}\left( r_i,\boldsymbol{\theta}_i | \boldsymbol{r}_{-i}, \boldsymbol{\theta}_{-i}\right) dr_i d\boldsymbol{\theta}_i d\boldsymbol{r}_{-i} d\boldsymbol{\theta}_{-i}
\end{split}\\
\begin{split}\label{eq:second term multi} 
&~+ 2 \sum_{i=1}^n \int_{\boldsymbol{\theta}_{-i}, \boldsymbol{r}_{-i}} \hat{g}(\boldsymbol{r}_{-i}, \boldsymbol{\theta}_{-i}) \int_{\boldsymbol{\theta}_i} I \{r^*_i(\boldsymbol{\theta},\boldsymbol{r}_{-i}) > \sqrt{m}(1+\delta)w^*_i(\boldsymbol{\theta}_{-i},\boldsymbol{r}_{-i}) \}  r^*_i(\boldsymbol{\theta},\boldsymbol{r}_{-i}) \\
&~~~~~~~~~~~~~~~~~~~~~~~~~~~~~~~~ \int_{r_i = \frac{r^*_i(\boldsymbol{\theta},\boldsymbol{r}_{-i})}{\sum_{k=1}^m \trig_k(\boldsymbol{\theta}_i)}}^{\infty} \hat{g}\left( r_i,\boldsymbol{\theta}_i | \boldsymbol{r}_{-i}, \boldsymbol{\theta}_{-i}\right) dr_i d\boldsymbol{\theta}_i d\boldsymbol{r}_{-i} d\boldsymbol{\theta}_{-i}. 
\end{split}
\end{talign}

Similarly to the two item case, bounding~\ref{eq:no good term multi} is problematic. To go around the issue we separate~\ref{eq:no good term multi} cases depending on whether $r^*_i(\boldsymbol{\theta},\boldsymbol{r}_{-i})$ is almost the second highest value, or at least $\sqrt{m}(1+\delta)$ times the second highest value. We get terms~\ref{eq:first term multi} and~\ref{eq:second term multi} which we bound separately.

\paragraph{Bounding~\ref{eq:first term multi}.}

We bound~\ref{eq:first term multi} by the revenue of a second price auction for the grand bundle. Notice that replacing $r^*_i(\boldsymbol{\theta},\boldsymbol{r}_{-i})$ with $w^*_i(\boldsymbol{\theta}_{-i},\boldsymbol{r}_{-i})$ makes the probability of sale increase, and the price decreases by a factor of at most $\frac{1}{\sqrt{m}(1+\delta)}$:

\begin{talign*}
\ref{eq:first term multi} &\leq 2 \sum_{i=1}^n \int_{\boldsymbol{\theta}_{-i}, \boldsymbol{r}_{-i}} \hat{g}(\boldsymbol{r}_{-i}, \boldsymbol{\theta}_{-i}) \int_{\boldsymbol{\theta}_i} (1+\delta)w^*_i(\boldsymbol{\theta}_{-i},\boldsymbol{r}_{-i}) \\
&~~~~~~~~~~~~~~~~~~~~~~~~~~~~~~~~~~~~ \int_{r_i = \frac{w^*_i(\boldsymbol{\theta}_{-i},\boldsymbol{r}_{-i})}{\sum_{k=1}^m \trig_k(\boldsymbol{\theta}_i)}}^{\infty} \hat{g}\left( r_i,\boldsymbol{\theta}_i | \boldsymbol{r}_{-i}, \boldsymbol{\theta}_{-i}\right) dr_i d\boldsymbol{\theta}_i d\boldsymbol{r}_{-i} d\boldsymbol{\theta}_{-i} \\
&=2\sqrt{m}(1+\delta) \sum_{i=1}^n \int_{\boldsymbol{\theta}_{-i}, \boldsymbol{r}_{-i}} \hat{g}(\boldsymbol{r}_{-i}, \boldsymbol{\theta}_{-i}) w^*_i(\boldsymbol{\theta}_{-i},\boldsymbol{r}_{-i}) \\
&~~~~~~~~~~~~~~~~~~~~~~~~~~~~~~~~~~~~ \int_{\boldsymbol{\theta}_i}  \int_{r_i = \frac{w^*_i(\boldsymbol{\theta}_{-i},\boldsymbol{r}_{-i})}{\sum_{k=1}^m \trig_k(\boldsymbol{\theta}_i)}}^{\infty} \hat{g}\left( r_i,\boldsymbol{\theta}_i | \boldsymbol{r}_{-i}, \boldsymbol{\theta}_{-i}\right) dr_i d\boldsymbol{\theta}_i d\boldsymbol{r}_{-i} d\boldsymbol{\theta}_{-i} \\
&= 2 \sqrt{m}(1+\delta) \secondpricerev{\wiggleD} \leq 2\sqrt{m} (1+\delta) \secondprev{\wiggleD}.
\end{talign*}

\paragraph{Bounding~\ref{eq:second term multi}.}

Our approach for bounding~\ref{eq:second term multi} will be similar to the single buyer case. Specifically, we first upper bound the inner most integral, for fixed $\boldsymbol{\theta}$ and $\boldsymbol{r}_{-i}$, by switching to cartesian coordinates, applying the formula to go from $\hat{f}$ to $f$ for our perturbation model, changing order of integration using Fubini's theorem, and then arguing about the size of the area that gets mapped to the point $r_i \sum_{k=1}^m \trig_k(\boldsymbol{\theta}_i)$.

\begin{talign*}
&\int_{r_i = \frac{r^*_i(\boldsymbol{\theta},\boldsymbol{r}_{-i})}{\sum_{k=1}^m \trig_k(\boldsymbol{\theta}_i)}}^{\infty} \hat{g}\left( r_i,\boldsymbol{\theta}_i | \boldsymbol{r}_{-i}, \boldsymbol{\theta}_{-i}\right) dr_i  \\
&= \int_{r_i = \frac{r^*_i(\boldsymbol{\theta},\boldsymbol{r}_{-i})}{\sum_{k=1}^m \trig_k(\boldsymbol{\theta}_i)}}^{\infty} r_i^{m-1} \prod_{k=1}^{m-2} (\sin{\boldsymbol{\theta}_{i,k}})^{m-k-1} \hat{f}\left( r_i \sum_{k=1}^m \trig_k(\boldsymbol{\theta}_i) | \boldsymbol{r}_{-i}, \boldsymbol{\theta}_{-i}\right) dr_i \\
&\leq \left( \frac{(1+\delta) \sqrt{m} }{\delta} \right)^m \prod_{k=1}^{m-2} (\sin{\boldsymbol{\theta}_{i,k}})^{m-k-1} \int_{r_i = \frac{r^*_i(\boldsymbol{\theta},\boldsymbol{r}_{-i})}{\sum_{k=1}^m \trig_k(\boldsymbol{\theta}_i)}}^{\infty} \frac{1}{r_i}  \int_{\mathbf{x}} f( \mathbf{x} | \boldsymbol{r}_{-i}, \boldsymbol{\theta}_{-i}) I\{ \mathbf{x} \in R(r_i,\boldsymbol{\theta}_i) \} d\mathbf{x} dr_i \\
& \leq \left( \frac{(1+\delta) \sqrt{m} }{\delta} \right)^m \prod_{k=1}^{m-2} (\sin{\boldsymbol{\theta}_{i,k}})^{m-k-1}  \int_{\mathbf{x}} f( \mathbf{x} | \boldsymbol{r}_{-i}, \boldsymbol{\theta}_{-i}) \int_{r_i = \frac{r^*_i(\boldsymbol{\theta},\boldsymbol{r}_{-i})}{\sum_{k=1}^m \trig_k(\boldsymbol{\theta}_i)}}^{\infty} \frac{I\{ \mathbf{x} \in R(r_i,\boldsymbol{\theta}_i) \}}{r_i} dr_i d\mathbf{x} \\
&\leq \left( \frac{(1+\delta) \sqrt{m} }{\delta} \right)^m \prod_{k=1}^{m-2} (\sin{\boldsymbol{\theta}_{i,k}})^{m-k-1}  \\
&~~~~~~~~~~~~~~~~~\int_{\mathbf{x}:\|\mathbf{x}\| \geq \frac{r^*_i(\boldsymbol{\theta},\boldsymbol{r}_{-i})}{(1+\delta)\sum_{k=1}^m \trig_k(\boldsymbol{\theta}_i)}} f( \mathbf{x} | \boldsymbol{r}_{-i}, \boldsymbol{\theta}_{-i}) \int_{r_i = \frac{r^*_i(\boldsymbol{\theta},\boldsymbol{r}_{-i})}{\sum_{k=1}^m \trig_k(\boldsymbol{\theta}_i)}}^{\infty} \frac{I\{ \mathbf{x} \in R(r_i,\boldsymbol{\theta}_i) \}}{r_i} dr_i d\mathbf{x} \\
&\leq \left( \frac{(1+\delta) \sqrt{m} }{\delta} \right)^m \prod_{k=1}^{m-2} (\sin{\boldsymbol{\theta}_{i,k}})^{m-k-1} \log(1+\delta)  \int_{\mathbf{x}:\|\mathbf{x}\| \geq \frac{r^*_i(\boldsymbol{\theta},\boldsymbol{r}_{-i})}{(1+\delta)\sum_{k=1}^m \trig_k(\boldsymbol{\theta}_i)}} f( \mathbf{x} | \boldsymbol{r}_{-i}, \boldsymbol{\theta}_{-i}) \\
&\leq \left( \frac{(1+\delta) \sqrt{m} }{\delta} \right)^m \prod_{k=1}^{m-2} (\sin{\boldsymbol{\theta}_{i,k}})^{m-k-1} \log(1+\delta)  Pr\left[ \text{$\mathbf{x} \sim \realD$ has $\|\mathbf{x}\|_1 \geq \frac{r^*_i(\boldsymbol{\theta},\boldsymbol{r}_{-i})}{(1+\delta)\sum_{k=1}^m \trig_k(\boldsymbol{\theta}_i)}$} \mid \boldsymbol{r}_{-i}, \boldsymbol{\theta}_{-i} \right] \\
&\leq \left( \frac{(1+\delta) \sqrt{m} }{\delta} \right)^m \prod_{k=1}^{m-2} (\sin{\boldsymbol{\theta}_{i,k}})^{m-k-1} \log(1+\delta)  Pr\left[ \text{$\mathbf{x} \sim \wiggleD$ has $\|\mathbf{x}\|_1 \geq \frac{r^*_i(\boldsymbol{\theta},\boldsymbol{r}_{-i})}{(1+\delta)\sqrt{m}}$} \mid \boldsymbol{r}_{-i}, \boldsymbol{\theta}_{-i} \right].
\end{talign*}

We can now bound~\ref{eq:second term multi}:
\begin{talign*}
&\ref{eq:second term multi} \leq 2 \left( \frac{(1+\delta) \sqrt{m} }{\delta} \right)^m \log(1+\delta) \sum_{i=1}^n \int_{\boldsymbol{\theta}_{-i}, \boldsymbol{r}_{-i}} \hat{g}(\boldsymbol{r}_{-i}, \boldsymbol{\theta}_{-i}) \int_{\boldsymbol{\theta}_i} \prod_{k=1}^{m-2} (\sin{\boldsymbol{\theta}_{i,k}})^{m-k-1} \\
&I \{\frac{r^*_i(\boldsymbol{\theta},\boldsymbol{r}_{-i})}{(1+\delta)\sqrt{m}} > w^*_i(\boldsymbol{\theta}_{-i},\boldsymbol{r}_{-i}) \} r^*_i(\boldsymbol{\theta},\boldsymbol{r}_{-i})  Pr\left[ \text{$\mathbf{x} \sim \wiggleD$ has $\|\mathbf{x}\|_1 \geq \frac{r^*_i(\boldsymbol{\theta},\boldsymbol{r}_{-i})}{(1+\delta)\sqrt{m}}$} \mid \boldsymbol{r}_{-i}, \boldsymbol{\theta}_{-i} \right] d\boldsymbol{\theta}_i \boldsymbol{\theta}_{-i}\boldsymbol{r}_{-i} \\
&= 2\sqrt{m}(1+\delta) \left( \frac{(1+\delta) \sqrt{m} }{\delta} \right)^m \log(1+\delta) \sum_{i=1}^n \int_{\boldsymbol{\theta}_{-i}, \boldsymbol{r}_{-i}} \hat{g}(\boldsymbol{r}_{-i}, \boldsymbol{\theta}_{-i}) \int_{\boldsymbol{\theta}_i} \prod_{k=1}^{m-2} (\sin{\boldsymbol{\theta}_{i,k}})^{m-k-1} \\
&I \{\frac{r^*_i(\boldsymbol{\theta},\boldsymbol{r}_{-i})}{(1+\delta)\sqrt{m}} > w^*_i(\boldsymbol{\theta}_{-i},\boldsymbol{r}_{-i}) \} \frac{r^*_i(\boldsymbol{\theta},\boldsymbol{r}_{-i})}{(1+\delta)\sqrt{m}}  Pr\left[ \text{$\mathbf{x} \sim \wiggleD$ has $\|\mathbf{x}\|_1 \geq \frac{r^*_i(\boldsymbol{\theta},\boldsymbol{r}_{-i})}{(1+\delta)\sqrt{m}}$} \mid \boldsymbol{r}_{-i}, \boldsymbol{\theta}_{-i} \right] d\boldsymbol{\theta}_i \boldsymbol{\theta}_{-i}\boldsymbol{r}_{-i} \\
&= 2\sqrt{m}(1+\delta) \left( \frac{(1+\delta) \sqrt{m} }{\delta} \right)^m \log(1+\delta) \sum_{i=1}^n \int_{\boldsymbol{\theta}_{-i}, \boldsymbol{r}_{-i}} \hat{g}(\boldsymbol{r}_{-i}, \boldsymbol{\theta}_{-i}) \int_{\boldsymbol{\theta}_i} \prod_{k=1}^{m-2} (\sin{\boldsymbol{\theta}_{i,k}})^{m-k-1} \\
&I \{q^*_i(\boldsymbol{\theta},\boldsymbol{r}_{-i}) > w^*_i(\boldsymbol{\theta}_{-i},\boldsymbol{r}_{-i}) \} q^*_i(\boldsymbol{\theta},\boldsymbol{r}_{-i})  Pr\left[ \text{$\mathbf{x} \sim \wiggleD:\|\mathbf{x}\|_1 \geq q^*_i(\boldsymbol{\theta},\boldsymbol{r}_{-i})$} \mid \boldsymbol{r}_{-i}, \boldsymbol{\theta}_{-i} \right] d\boldsymbol{\theta}_i \boldsymbol{\theta}_{-i}\boldsymbol{r}_{-i} \\
&\leq 2\sqrt{m}(1+\delta) \left( \frac{(1+\delta) \sqrt{m} }{\delta} \right)^m \log(1+\delta) \sum_{i=1}^n \int_{\boldsymbol{\theta}_{-i}, \boldsymbol{r}_{-i}} \hat{g}(\boldsymbol{r}_{-i}, \boldsymbol{\theta}_{-i}) \int_{\boldsymbol{\theta}_i} \prod_{k=1}^{m-2} (\sin{\boldsymbol{\theta}_{i,k}})^{m-k-1} \\
&I \{q^*_i(\boldsymbol{\theta}_{-i},\boldsymbol{r}_{-i}) > w^*_i(\boldsymbol{\theta}_{-i},\boldsymbol{r}_{-i}) \} q^*_i(\boldsymbol{\theta}_{-i},\boldsymbol{r}_{-i})  Pr\left[ \text{$\mathbf{x} \sim \wiggleD:\|\mathbf{x}\|_1 \geq q^*_i(\boldsymbol{\theta}_{-i},\boldsymbol{r}_{-i})$} \mid \boldsymbol{r}_{-i}, \boldsymbol{\theta}_{-i} \right] d\boldsymbol{\theta}_i \boldsymbol{\theta}_{-i}\boldsymbol{r}_{-i} \\
&\leq 2 \sqrt{m}(1+\delta) \left( \frac{(1+\delta) \sqrt{m} }{\delta} \right)^m \log(1+\delta) \cdot \left(  \left( \sqrt{\frac{\pi e}{2}} \right)^{m} \frac{m}{(\sqrt{m})^{m}}\right) \cdot \secondprev{\wiggleD} \\
&= 2 \left( \sqrt{\frac{\pi e}{2}} \frac{(1+\delta) }{\delta} \right)^m (1+\delta) \log(1+\delta)  m \sqrt{m} \cdot \secondprev{\wiggleD}.
\end{talign*}

\paragraph{Putting it all together.} 
Combining our bound for~\ref{eq:first term multi} and~\ref{eq:second term multi} we get:

\[
\rev{\wiggleD} \leq 4 \left( \sqrt{\frac{\pi e}{2}} \frac{(1+\delta) }{\delta} \right)^m (1+\delta) \log(1+\delta)  m \sqrt{m} \cdot \secondprev{\wiggleD}. \qedhere
\]

\end{proof}

%% file: extensions.tex
\section{Extensions}\label{sec:extensions}

It is straightforward to extend our results beyond additive valuations. 

\subsection*{Unit-Demand Valuations}

Our bounds for unit-demand valuations are the same as the ones for additive valuations, in both the Angle-Shift and the Square-Shift models. Starting from our toy Angle-Shift model, notice that in the proof of Theorem~\ref{thm:angle theorem new} Lemma~\ref{lemma:anglefocus} can still be applied, but $\wiggleD_{\theta}$ can be replaced by the single parameter distribution where we sample $r$ as before and output $\max\{r\cos \theta,r \sin \theta\}$ (instead of $r(\cos\theta + \sin\theta)$). This would give  a term $\int_{r = \frac{r_\theta}{\max\{\sin{\theta},\cos{\theta}\}}}^{\infty} g(r) dr$, which is equal to the probability that a sample from $\realD$ has length at least $\frac{r_\theta}{\max\{\sin{\theta},\cos{\theta}\}}$, which is equal to the probability of the same event for $\wiggleD$, which is equal to $Pr\left[\text{$(x,y) \sim \wiggleD$ has $\max\{x,y\} \geq r_{\theta}$}\right]$. The latter probability multiplied by $r_\theta$ is upper bounded by $\brev{\wiggleD}$. Therefore, for this setting, we overall get the same bound of $\frac{\pi \sqrt{2}}{2 \delta} \brev{\wiggleD}$ for unit-demand valuations.

In the Square-Shift model, starting from a single buyer and two items, the same observation about $\wiggleD_{\theta}$ holds. When bounding 
$\int_{r \max\{\cos\theta,\sin\theta\} = r_\theta}^{\infty}  \hat{g}(r, \theta) dr$, Claim~\ref{claim:a b length} is replaced by:

\begin{claim}\label{claim:a b length unit demand}
$\forall (a,b) \in R(r,\theta)$, if $r \geq \frac{r_\theta}{\max\{\cos{\theta},\sin{\theta}\}}$ then $\max\{a , b\} \geq \frac{r_\theta}{1+\delta}$.
\end{claim}

\begin{proof}
Recall that $(a,b) \in R(r,\theta)$ means that $(a,b)$ that could map to a point with length $r$ and angle $\theta$.
Then, $(1+\delta) \max\{a,b\} \geq \max\{ r \cos\theta, r\sin\theta \} \geq r_\theta$. The Claim follows.
\end{proof}

Claim~\ref{claim:singledimintbound} remains unchanged. Overall we get:

\begin{align*}
&\int_{r=\frac{r_\theta}{\max\{\cos{\theta},\sin{\theta}\}}}^{\infty}  \hat{g}(r, \theta) dr \leq \frac{2(1+\delta)^2 \log{(1+\delta)}}{\delta^2}  \int_{a = 0}^{\infty} \int_{b = \frac{r_\theta}{1+\delta} - a}^{\infty} f(a,b) db da \\
&\quad= \frac{2(1+\delta)^2 \log{(1+\delta)}}{\delta^2} Pr\left[ \text{sample from $\realD$ has maximum value at least $\frac{r_\theta}{1+\delta}$} \right] \\
&\quad \leq \frac{2(1+\delta)^2 \log{(1+\delta)}}{\delta^2} Pr\left[ \text{sample from $\wiggleD$ has maximum value at least $\frac{r_\theta}{1+\delta}$} \right].
\end{align*}

Therefore, the bound $\rev{\wiggleD} \leq \frac{\sqrt{2} \pi (1+\delta)^3 \log{(1+\delta)}}{\delta^2} \brev{\wiggleD}$ is valid a unit demand buyer and two items in the Square-Shift model. For a single buyer and $m$ items we reach the same conclusion by replacing Claim~\ref{claim:x length general box} appropriately. Similarly for the many buyer case two item, and many buyer multi-item cases:~\ref{eq:first term} is bounded by a second price auction for the grand bundle (where this time  buyer $i$'s value for the grand bundle is $\max\{v_{i,1},v_{i,2}\}$ and not $v_{i,1} + v_{i,2}$), and the second term (\ref{eq:second term}) is bounded similarly to the single buyer case.

\subsection*{Additive subject to downward-closed constraints $\mathcal{I}$}

Let $\rev{\realD}$ be the revenue of the optimal truthful mechanism on an $m$-dimensional, $n$ (correlated) buyer distribution $\realD$, where the utility of buyer $i$ for a subset of the items $S$ is equal $v_i(S) = \max_{T \subseteq S, T \in \mathcal{I}} \{ \sum_{j \in T} v_{i,j} \}$. Also, let $\arev{\realD}$ be the optimal revenue of the same distribution, but this time buyers have additive valuation functions.

\begin{claim}
$\rev{\realD} \leq \arev{\realD}$.
\end{claim}

Therefore, we can apply our approximations for the additive setting. It remains to observe that the revenue of bundling under additive valuations is at most $m$ times the revenue of bundling under downward-closed constraints $\mathcal{I}$, i.e. we lose a factor $m$ in the approximation guarantee.

%% file: conclusion.tex
\section{Conclusion and Future Directions}\label{sec:conclusions}
We initiate the smoothed analysis of multi-item auctions with arbitrarily correlated buyers. We present two main results: first, the~\cite{BriestCKW15, HartN13} construction is surprisingly robust, and the infinite gap persists in our ``Rectangle-Shift'' model. Second, we show that a smoothed-finite approximation is indeed possible in our ``Square-Shift'' model, with extensions to multiple buyers and multiple items. As a bonus technical insight, we learn that whether or not an infinite gap exists between simple and optimal for some distribution $\mathcal{D}$ is intimately connected to what $\mathcal{D}$ looks like in small cones.

Our work takes the view of studying truly simple mechanisms whose approximation ratios degrade with the smoothing parameter. An alternative approach would be to study truly good approximation ratios via mechanisms whose complexity degrades with the smoothing parameter. To be a little more concrete, an interesting open question is the following: in the Square-Shift model, what is the required \emph{menu complexity} of an auction (that is, the number of possible outcomes based on the buyers' input) to guarantee a $(1-\varepsilon)$-approximation if $\wiggleD$ is smoothed with parameter $\delta$?

A general direction for future work is exploring alternative smoothed models. As an example, for many items our guarantees are exponential in $m$, and this is tight for our model. Would an alternative natural smoothed model provide an improved guarantee, or is there some inherent barrier?

The ``independent items'' paradigm led to numerous developments which deepened our understanding of multi-item auctions. This work shows that smoothed analysis can help guide future work to better understand worst-case distributions beyond the impossibility results of~\cite{BriestCKW15, HartN13}.

%% file: table_appendix.tex
\section{Results and Figures}\label{app:table}

\begin{center}
\resizebox{\columnwidth}{!}{%
\begin{tabular}{|c|c|c||c|}
\hline
  & Independent & Correlated & Square-Shift \\
\hline\hline
$n=1, m=2$ & $\frac{\rev{\realD}}{\brev{\realD}} \leq 2$~\cite{HartN12} &  $\frac{\rev{\realD}}{\brev{\realD}}$ unbounded~\cite{HartN13} & $\frac{\rev{\wiggleD}}{\brev{\wiggleD}} \leq \frac{\sqrt{2} \pi (1+\delta)^3 \log{(1+\delta)}}{\delta^2}$ \\ \hline
$n=1$, general $m$ &  $\frac{\rev{\realD}}{\max\{\srev{\realD},\brev{\realD}\}} \leq 6$~\cite{BabaioffILW14} & &  $\frac{2^m}{m} \leq \frac{\rev{\wiggleD}}{\brev{\wiggleD}} \leq \left( \sqrt{\frac{\pi e}{2}} \frac{(1+\delta) }{\delta} \right)^m (1+\delta) \log(1+\delta)  m \sqrt{m}$  \\ \hline
general $n$, $m=2$ & $\frac{\drev{\realD}}{\srev{\realD}} \in O(1)$~\cite{BabaioffILW14, Yao15} & &  $\frac{\drev{\wiggleD}}{\secondprev{\wiggleD}} \leq \frac{ 9\sqrt{2}(1+\delta)^3 \log(1+\delta)}{ \delta^2 }$ \\ \hline
general $n, m$ & $\frac{\BICrev{\realD}}{\max\{\srev{\realD}, \brev{\realD}\}} \leq 8$~\cite{Yao15,CaiDW16} &  & $\frac{\drev{\wiggleD}}{\secondprev{\wiggleD}} \in O\left(\frac{(1+\delta)}{\delta}\right)^m (1+\delta) \log(1+\delta) m \sqrt{m} $  \\ \hline
\end{tabular}
}
\end{center}


\input{fancy_model_figures}

\section{Extended Preliminaries}\label{app:prelim}
For the sake of completion, we provide a brief discussion on the BIC solution concept for correlated buyers. Note that when buyers are correlated, your prior on the other buyers' values changes depending on your own prior. So Bayesian IC truly means ``you prefer to tell the truth, assuming that all other buyers tell the truth (even taking your updated prior into account).'' Formally, one needs to define $\pi_{ij}(\mathbf{v}_i, \mathbf{v}'_i)$ to be the probability that buyer $i$ receives item $j$ conditioned on reporting $\mathbf{v}_i$ while their true type is $\mathbf{v}'_i$. When buyers are independent the true type doesn't affect the conditioning. And then a mechanism is BIC as long as:

$$\mathbf{v}_i \cdot \vec{\pi}_i (\mathbf{v}_i,\mathbf{v}_i) - p_i(\mathbf{v}_i) \geq \mathbf{v}_i \cdot \mathbf{\pi}_i(\mathbf{v}_i,\mathbf{v}'_i) - p_i(\mathbf{v}'_i).$$

Also, while the difference between ex-post IR (always receive non-negative utility for reporting the truth) and interim IR (receive non-negative utility for reporting the truth in expectation) is unimportant for auctions with independent buyers (due to a very simple reduction from interim IR to ex-post IR, see e.g.~\cite{DaskalakisW12}). There's a fundamental difference for correlated buyers. This is because seminal work of Cremer and McLean~\cite{CremerM88} shows that for any non-generic correlated distribution, it is possible for the seller to extract the full welfare with a DSIC, interim IR mechanism (but this is not possible with a DSIC, ex-post IR mechanism). So all of our results stated for correlated buyers necessarily reference ex-post IR mechanisms (and this is explicitly stated in the body). 

%% file: fancy_model_figures.tex
\begin{figure}[ht]
\begin{subfigure}{.26\linewidth}
\centering
\begin{tikzpicture}[scale=0.35]

  
    \coordinate (Origin)   at (0,0);
    \coordinate (XAxisMin) at (0,0);
    \coordinate (XAxisMax) at (10,0);
    \coordinate (YAxisMin) at (0,0);
    \coordinate (YAxisMax) at (0,7);
    \draw [thin, black,-latex] (XAxisMin) -- (XAxisMax);
    \draw [thin, black,-latex] (YAxisMin) -- (YAxisMax);

\node[draw,circle,inner sep=1pt,fill] at (0,0) {};
\node[draw,circle,inner sep=1pt,fill] at (4,2) {};
\node [below left] at (4,2) {$(x,y)$};
\node[draw,circle,inner sep=1pt,fill] at (8,2) {};
\node[draw,circle,inner sep=1pt,fill] at (4,4) {};
\node[draw,circle,inner sep=1pt,fill] at (8,4) {};

\draw[decorate,decoration={brace,amplitude=5pt,mirror}] 
    (4,1.8) -- (8,1.8);
\node [below] at (6,1.5) {$\delta x$};
\draw[decorate,decoration={brace,amplitude=5pt,mirror}] 
    (3.8,4) -- (3.8,2);
\node [left] at (3.3,3.5) {$\delta y$};

\draw [dashed] (4,2) -- (8,2) -- (8,4) -- (4,4) -- (4,2);
\draw [pattern=north east lines, pattern color=blue] (4,2) -- (8,2) -- (8,4) -- (4,4) -- (4,2);

\end{tikzpicture}
\caption{Rectangle-Shift.}
\label{fig:sub3}
\end{subfigure}\hspace{40px}
\begin{subfigure}{.26\linewidth}
\centering
\begin{tikzpicture}[scale=0.35]

  
    \coordinate (Origin)   at (0,0);
    \coordinate (XAxisMin) at (0,0);
    \coordinate (XAxisMax) at (10,0);
    \coordinate (YAxisMin) at (0,0);
    \coordinate (YAxisMax) at (0,7);
    \draw [thin, black,-latex] (XAxisMin) -- (XAxisMax);
    \draw [thin, black,-latex] (YAxisMin) -- (YAxisMax);

\node[draw,circle,inner sep=1pt,fill] at (0,0) {};
\node[draw,circle,inner sep=1pt,fill] at (4,2) {};
\node [below left] at (4,2) {$(x,y)$};
\node[draw,circle,inner sep=1pt,fill] at (7,2) {};
\node[draw,circle,inner sep=1pt,fill] at (4,5) {};
\node[draw,circle,inner sep=1pt,fill] at (7,5) {};

\draw[decorate,decoration={brace,amplitude=5pt,mirror}] 
    (4,1.8) -- (7,1.8);
\node [below] at (5.5,1.5) {$\delta x$};
\draw[decorate,decoration={brace,amplitude=5pt,mirror}] 
    (3.8,5) -- (3.8,2);
\node [left] at (3.3,3.5) {$\delta x$};

\draw [dashed] (4,2) -- (7,2) -- (7,5) -- (4,5) -- (4,2);
\draw [pattern=north east lines, pattern color=blue] (4,2) -- (7,2) -- (7,5) -- (4,5) -- (4,2);

\end{tikzpicture}
\caption{Square-Shift.}
\label{fig:sub2}
\end{subfigure}\hspace{40px}
\begin{subfigure}{.26\linewidth}
\centering
\begin{tikzpicture}[scale=0.35]
    \coordinate (Origin)   at (0,0);
    \coordinate (XAxisMin) at (0,0);
    \coordinate (XAxisMax) at (10,0);
    \coordinate (YAxisMin) at (0,0);
    \coordinate (YAxisMax) at (0,7);
    \draw [thin, black,-latex] (XAxisMin) -- (XAxisMax);
    \draw [thin, black,-latex] (YAxisMin) -- (YAxisMax);

\node[draw,circle,inner sep=1pt,fill] at (0,0) {};
    
\node[draw,circle,inner sep=1pt,fill] at (6,5) {};
\node [above right] at (6,5) {$(r,\theta)$};
\draw [->|] (0,0) -- node[anchor=north,rotate=30]{$r$} (6,5);

\draw [dashed] (2,0) -- ([shift={(0,0)}]0:2)arc[radius = 2, start angle= 0, end angle=40];
\node [right] at (2,0.8) {$\theta$};

\draw [thick] (6,5) -- ([shift={(0,0)}]40:7.8)arc[radius = 7.8, start angle= 40, end angle=26];
\draw [thick] (6,5) -- ([shift={(0,0)}]40:7.8)arc[radius = 7.8, start angle= 40, end angle=54];

\node[draw,circle,inner sep=1pt,fill] at (4.65,6.25) {};
\node[draw,circle,inner sep=1pt,fill] at (7.05,3.35) {};
\node [left] at (4.65,6.25) {$(r,\theta+\delta)$};
\node [right] at (7.05,3.35) {$(r,\theta-\delta)$};

\end{tikzpicture}
\caption{Angle-Shift}
\label{fig:sub1}
\end{subfigure}
\caption{Our models.}
\label{fig:test}
\end{figure}
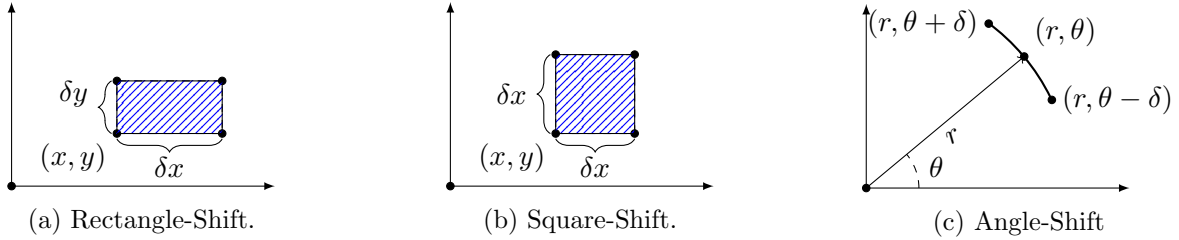

%% file: lb_appendix.tex
\section{Lower bounds}\label{appendix:lower bounds}

\subsection{Rectangle-Shift lower bounds}\label{app: rectangle lower bounds}

\begin{proof}[Proof of Lemma~\ref{lemma:lb auction}]
Let $M_i = \frac{4^i}{\prod_{j<i} gap^j_\delta}$.
Let $\realD$ be the following distribution: the buyer has value $\mathbf{v}_i = M_i \mathbf{x}_i$ with probability $\frac{1}{M_i}$. With the remaining probability $1 - \sum_{i=1}^\infty M_i$ the value is zero. 
We sample from $\wiggleD$ by first sampling $\mathbf{x}_i$ w.p. $\frac{1}{M_i}$, perturbing with parameter $\delta$ to $\hat{\mathbf{x}}_i$, and outputting value $\mathbf{v} = M_i \hat{\mathbf{x}}_i$.

Consider the following menu. The $i$-th menu item costs $M_i \text{gap}^i_\delta$ for an allocation of $\mathbf{x}_i$. We first show that every buyer $\hat{\mathbf{x}}_i \in R^{-1}(\mathbf{x}_i,\delta)$ prefers the $i$-th menu item. First, notice that the utility of such a buyer for the $i$-th item is $M_i \hat{\mathbf{x}}_i \mathbf{x}_i - M_i \text{gap}^i_\delta \geq M_i \hat{\mathbf{x}}_i \mathbf{x}_j$, for all $j<i$ (from the definition of $\text{gap}^i_\delta$), which is the utility from getting $j$-th menu item without any payment. For $j > i$, the utility of buyer $\hat{\mathbf{x}}_i$ for the $j$-th menu item is: $M_i \hat{\mathbf{x}}_i \mathbf{x}_j - M_j gap^j_\delta \leq 2 (1 + \delta) M_i - 4 M_{j-1} < 0$, since $j-1 \geq i$ and $\delta < 1/2$. Finally, the revenue of this auction is $\sum_{i=1}^\infty M_i \text{gap}^i_\delta \frac{1}{M_i} = \sum_{i=1}^\infty \text{gap}^i_\delta$.

It remains to show that $\brev{\wiggleD} \in O(1)$. For any point $\hat{\mathbf{x}}_i$, the buyer's value (and seller's revenue) for the grand bundle is at most $2 M_i (1+\delta)$ (post perturbing). Since $2M_{i-1} (1+\delta) < M_{i}$, the only types willing to buy the bundle at price $M_i \hat{\mathbf{x}}_i$ are those $j \geq i$. Then the probability of selling is at most $\sum_{j \geq i} \frac{1}{M_i} \leq \frac{2}{M_i}$ (since the series of $M_i$ is dominated by a geometric sequence with rate of growth of 2). Therefore the revenue of selling the grand bundle at any price $M_i\hat{\mathbf{x}}_i$ is at most $\frac{2}{M_i} 2 M_i (1+\delta) = 4(1+\delta)$.  
\end{proof}

\begin{proof}[Proof of Claim~\ref{claim:ezpz}]

Recall that the $k^{th}$ point on shell $N$ has angle $\frac{\pi}{2} \cdot (1-3\delta)^{k-1}$. Observe that for any $j > k$, we have $\theta_j = (1-3\delta)^{j-k} \cdot \theta_k \leq (1-3\delta)\theta_k$. In particular, this implies that $\theta_k - \theta_j \geq 3\delta \theta_k$. 

For the other direction, if $j < k$, we have just shown that $|\theta_k - \theta_j| \geq 3\delta \theta_j \geq 3\delta \theta_k$, since $\theta_k$ is a decreasing sequence. 
\end{proof}

\begin{proof}[Proof of Claim~\ref{claim:lb on theta}]
The most extreme angles $\hat{\theta}_i$ a point $p_i$ can be mapped to are $\arctan(\tan(\theta_i (1+\delta)))$ and $\arctan(\tan(\frac{\theta_i}{1+\delta}))$. We now bound how big the gap can be, as a function of $\theta_i$. We do so by examining the integral of the derivative of the function $\arctan((1+x)\tan(\theta_i))$ with respect to $x$.

\begin{align*}
\hat{\theta}_i-\theta_i &\leq \arctan((1+\delta)\tan(\theta_i))-\arctan \left(\tan \left(\theta_i \right)\right) \\ 
&\leq \int_{x=0}^{x=\delta} \frac{\tan(\theta_i)}{(\tan(\theta_i)+x\tan(\theta_i))^2+1}dx \\ 
&\leq \delta \frac{\tan(\theta_i)}{\tan(\theta_i)^2+1} \\
&= \delta \sin \theta_i \cos \theta_i \leq \delta \theta_i.\qedhere
\end{align*}

Similarly, we have (by examining the function $\arctan((1+x)\tan(\theta_i)/(1+\delta))$:
\begin{align*}
{\theta_i}-\hat{\theta}_i &\leq \arctan(\tan(\theta_i))-\arctan \left(\tan \left(\theta_i \right)/(1+\delta)\right) \\ 
&\leq \int_{x=0}^{x=\delta} \frac{\tan(\theta_i)/(1+\delta)}{(\frac{\tan(\theta_i)+x\tan(\theta_i)}{1+\delta})^2+1}dx \\ 
&\leq \delta \frac{(1+\delta)\tan(\theta_i)}{\tan(\theta_i)^2+(1+\delta)^2} \\
&\leq \delta \frac{(1+\delta) \sin \theta_i \cos \theta_i}{1+(2\delta-\delta^2)\cos^2\theta_i} \leq \delta \frac{\theta_i + \delta \theta_i \cos^2\theta_i}{1+\delta \cos^2 \theta_i} = \delta \theta_i.\qedhere
\end{align*}

The final line is due to the fact that $\theta_i \geq \sin \theta_i \cos \theta_i$ on $[0,\pi/2]$, and also that $\theta_i \cos \theta_i \geq \sin \theta_i$ on $[0,\pi/2]$ (and that $2\delta - \delta^2 \geq \delta$).\qedhere

\begin{proof}[Proof of Claim~\ref{claim:gap in omega}]
\begin{talign*}
&\mathbf{\hat{x}_i} \left( \mathbf{x}_i - \mathbf{x}_{j} \right)= \|\mathbf{\hat{x}}_i  \| \| \mathbf{x}_i \| \cos \left(  \theta_i - \hat{\theta}_i \right) - \|\mathbf{\hat{x}}_i\| \|\mathbf{x}_j\| \cos \left( \hat{\theta}_i - \theta_j \right)  \\
&= \|\mathbf{\hat{x}}_i  \| \ell_N \cos \left( \theta_i - \hat{\theta}_i \right) - \|\mathbf{\hat{x}}_i\| \ell_N \cos \left(\hat{\theta}_i - \theta_j \right) \\
&\geq \ell^2_N \left( \cos \left(  \theta_i - \hat{\theta}_i \right) -  \cos \left( \hat{\theta}_i - \theta_j \right) \right)~~~~~~~~~~~~~\left[ \| \mathbf{\hat{x}}_i \| \geq \ell_N \right] \\
&=  \ell^2_N \left(2 \sin \left(  \frac{|\theta_i - \theta_j|}{2} \right) \sin \left( \frac{|2\hat{\theta}_i - \theta_i - \theta_j |}{2} \right) \right)~~~~~~~~~~~\left[ \cos{x} - \cos{y} = 2 \sin{\frac{y+x}{2}} \sin{\frac{y-x}{2}},\ \ \ \sin x = \sin |x| \right]\\
&\geq  \frac{2}{\pi^2} \ell^2_N \left( |\theta_i-\theta_j| \right) \left( |2\hat{\theta}_i - 2\theta_i + \theta_i - \theta_j|\right)~~~~~~~~\left[ \text{Jordan's inequality:} \sin x \geq \frac{2}{\pi} x \right]\\
&\geq  \frac{2}{\pi^2} \ell^2_N \left( |\theta_i - \theta_j| - 2 |\hat{\theta}_i-\theta_i|  \right)^2 \\
&\geq \frac{2}{\pi^2} \ell^2_n \left(3 \delta \theta_i - 2 \delta \theta_i\right)^2 .~~~~~~~~~~~~~~~~~~~~~~~~~~~~\left[\text{Claims~\ref{claim:lb on theta} and~\ref{claim:ezpz} }\right]
\end{talign*}
\end{proof}

\end{proof}

\subsection{Square-Shift Lower Bound for Two Items}\label{app:lower bound square two}

\begin{lemma}
There exists a bivariate distribution $\realD$  such that for its corresponding perturbed distribution $\wiggleD$, $\rev{\wiggleD} \in \Omega\left( (\frac{1}{\delta})^{1/7} \right) \brev{\wiggleD}$.
\end{lemma}

\begin{proof}
First, we restate some known results that we will need.\footnote{These propositions appear in the arXiv version of~\cite{HartN13} uploaded on 22 April, 2013.}

Let $q^1, q^2, \dots$ be a finite or countably infinite sequence of points in $[0, 1]^k$. Define
\[
gap^n = q^n \cdot q^n - \max_{j < n} q^j \cdot q^n.
\]

\begin{proposition}[Proposition $5.1$ in~\cite{HartN13}]
For every (finite or countably infinite) sequence $q^1, q^2, \dots$ of points in $[0, 1]^k$ there exists a distribution $F$ on $\mathbb{R}^k_+$ such that $\brev{F} \leq 2k$, $\rev{F} \geq \sum_{n} gap^n$ and $\mathbb{E}[v_1 + v_2] = \sum_{n} \| q^n \|_1$. Specifically, for each $n$, let $M^n = (2k)^n/(\prod_{j=1}^n gap^j )$. Then, $F$ is the distribution that puts probability $1/M^n$ on the point $x^n = M^n q^n$, and puts the remaining probability on the point $0$.
\end{proposition}

\begin{proposition}[Proposition $6.1$ in~\cite{HartN13}]
There exists an infinite sequence of points $q^1, q^2, \dots$ of points in $[0,1]^2$ with $\| q^k \|_2 \leq 1$, such that taking for all $k$, $gap^k \in \Omega(k^{-6/7})$. The construction puts points in ``shells''. The $N$-th shell has $N^{3/4}$ points with length $\sum_{\ell = 1}^N \ell^{-3/2} / (\sum_{\ell = 1}^{\infty} \ell^{-3/2})$.
\end{proposition}

Combining these two results we can get a distribution $\realD$ over $n$ points, such that $\brev{\realD} \leq 4$ and $\rev{\realD} \geq \sum_{k=1}^n gap^k = \sum_{k=1}^n k^{-6/7} \in \Omega( n^{1/7} )$. The largest possible value for any of the two items is $M^n q^n \leq 2 M^n \leq  2 \frac{4^n}{\prod_{j=1}^n gap^j } \leq 4^{n+1}$. The expected welfare of $\realD$ is $\sum_{k=1}^n \| q^k \|_1 \in \Theta (  \sum_{k=1}^n \| q^k \|_2 ) = \Theta( \sum_{N=1}^{N_{max}} N^{3/4} \frac{\sum_{\ell = 1}^N \ell^{-3/2}}{\sum_{\ell = 1}^{\infty} \ell^{-3/2}} ) = \Theta( N_{max}^{7/4} )$, where $N_{max}$ is the number of shells. The total number of points is $n$, and the $N$-th shell has $N^{3/4}$ points, therefore, $N_{max} \in \Theta( n^{4/7} )$, and therefore the expected welfare is in $\Theta(n)$.

To complete the proof of the lemma, we find a $\delta$ that is large, but at the same time the gap between $\rev{\wiggleD}$ and $\brev{\wiggleD}$  remains sufficiently large as well. We use the following theorem:

\begin{theorem}[\cite{RubinsteinW15}]
Let $D$ be a distribution, and $D^+$ be a distribution that stochastically dominates $D$. Let $\Delta$ denote the random function $v^+ - v$, when couples $v^+$ and $v$ are sampled jointly from $D^+$ and $D$. Finally, let $\textsc{Val}(\Delta)$ be the expected welfare of VCG with buyers whose types are distributed according to $\Delta$. Then, for all such $D$, $D^{+}$, $\Delta$, and for every mechanism $M$ for $D$ and every $\epsilon > 0$, there exists a mechanism $M'$ for $D^{+}$ such that
\[
\rev{ M', D^{+} } \geq (1-\epsilon) \cdot \left( \rev{ M, D} - \frac{\textsc{Val}(\Delta)}{\epsilon} \right).
\]
\end{theorem}

We apply this theorem for $D = \realD$, and $D^{+} = \wiggleD$, which stochastically dominates $\realD$, and $M$ the optimal mechanism for $\realD$. Observe that $\textsc{Val}(\Delta) = \mathbb{E}[ \sum_{i=1}^2 (\hat{v_i} - v_i) ] \leq 2 \delta \mathbb{E}[ \max_i v_i ] \leq 2 \delta \mathbb{E}[ v_1 + v_2 ] \in \Theta( \delta n )$. We set $\epsilon = \frac{1}{2}$. Then:
\begin{align*}
\rev{\wiggleD} &\geq \rev{ M', \wiggleD } \\
&\geq \frac{1}{2} \cdot \left( \rev{\realD} - \Theta( \delta n  ) \right) \\
&\geq \frac{n^{1/7}}{2} - \Theta( \delta n  ).
\end{align*}
Picking $n \in \Theta( \frac{1}{\delta} )$ gives $\rev{\wiggleD}  \in \Omega( n^{1/7} )$, while  $\brev{\wiggleD} \leq \brev{\realD} + \delta n \in O(1)$.
\end{proof}

\subsection{Square-Shift Lower Bound for Multiple Items}\label{app:lower bound square}

In this subsection we present a lower bound for the Square-Shift model with a single buyer interested in multiple items. 

\begin{restatable}{thm}{tightUB}
For small enough $\delta$, there exists a distribution over $m$ items $\realD$ such that $\brev{\wiggleD} \in O(m^2)$ but $\rev{\wiggleD} \in \Omega\left(m 2^m\right)$.
\end{restatable}

\begin{proof}
Enumerate all subsets of $[m]$ of size $\frac{m}{2}$ in some order: $S_1, S_2, ... , S_M$, $M = \binom{m}{m/2}$, and let $C = \sum_{j=1}^{M} m^{-2j} = \frac{1}{m^2 - 1}$. For $\realD$, the buyer will have type $j$ with probability $m^{-2j}/C$. A buyer of type $j$ is interested in each item in $S_j$ equally, for a value of $2m^{2j}$ per item, and is not interested in items not in $S_j$. 

The valuations in $\wiggleD$ are not significantly different. The buyer with type $j$ has a value of at most $\delta 2 m^{j}$ for an item not in $S_j$ and $(1+\delta)2m^{2j}$ for an item in $S_j$. This means that for the buyer's value for the grand bundle is at most $\frac{m}{2} \delta 2m^{2j} + \frac{m}{2} (1+\delta) 2m^{2j} \leq m^{2j+1}(1+2\delta)$, which is smaller than $2m^{2(j+1)}$. Therefore, if the grand bundle is priced at $2m^{2j}$, only buyers with types $k \geq j$ will purchase. Therefore, the revenue of bundling is $\brev{\wiggleD} \leq (1+2\delta)2m^{2j} \sum_{i \geq j}^{k} \frac{m^{-2i}}{C} \leq 4(1+2\delta) \frac{1}{C} \in O(m^2)$.

Consider the following mechanism. Offer the set $S_j$ at a price $m^{2j}$. We show that a buyer with type $j$ always prefers that tailored offer to any other subset. For $k > j$, the cost of $S_k$ is at least $m^{2j+2}$. The value of $j$ for $S_j$ is at most $\frac{m}{2} 2m^{2j} (1+\delta) = m^{2j+1}(1+\delta)$, which is less than the cost of $S_k$ for $k > j$. For $k<j$, the utility of a type $j$ for getting $S_k$ for free is at most $V_k = \sum_{i=1}^{m-1} \hat{v}_i$, where $\hat{v}_i$ is the value for the $i$-th valuable item after perturbing.  The utility of $j$ for $S_j$ is at least $\sum_{i=1}^{m-1} \hat{v}_i + 2m^{2j} - m^{2j} = V_k + m^{2j}$. Therefore, type $j$ purchases $S_k$ for a price of $m^{2j}$. The total revenue is $\sum_{j=1}^k m^{2j} \frac{m^{-2j}}{C} = \frac{k}{C} \geq m 2^m$.
\end{proof}

%% file: ub_appendix.tex
\section{Upper Bounds}\label{app:all upper bounds}
\subsection*{Polar Coordinates} 

It will turn out extremely useful for our analysis to look at valuations in polar coordinates as opposed to Cartesian coordinates. We use this space to remind the reader about some useful properties. Given a $m$-dimensional vector $\mathbf{x}$, its polar transformation is $(r, \boldsymbol{\theta})$, where $r$ is a non-negative real number and $\boldsymbol{\theta}$ is an $m-1$ dimensional vector of angles $\theta_i \in [0, \frac{\pi}{2}]$ such that $x_1 = r \cos \theta_1$, $x_2 = r \sin \theta_1 \cos \theta_2$, $x_3 = r \sin \theta_1 \sin \theta_2 \cos \theta_3$ and so on until $x_{m-1} = r \prod_{k=1}^{m-2} \sin \theta_k \cos \theta_{m-1}$, $x_{m} = r \prod_{k=1}^{m} \sin \theta_k$. Let $\trig \left( \boldsymbol{\theta} \right)$ be a vector such that $\trig_j \left( \boldsymbol{\theta} \right) = \cos{\theta_j} \prod_{k=1}^{j-1} \sin{\theta_k}$. We sometimes write $\left( r \cdot \trig \left( \boldsymbol{\theta} \right) \right)$ or $\left( r, \trig \left( \boldsymbol{\theta} \right) \right)$ for the transformation from polar to cartesian coordinates.

\begin{claim}
\label{claim:trigfact}
$\max_{j=1,\dots,m}{(\trig_j(\boldsymbol{\theta}))} \geq \frac{1}{\sqrt{m}}.$
\end{claim}

\begin{proof}
$\sum_{j=1}^m \trig(\boldsymbol{\theta})_j^2 = 1$ and therefore the largest entry (of the sum) must be greater than their average, $1/m$. Since the entries are non-negative we can take a square root; the Claim follows. 
\end{proof}  

Let $g_i$ and $\hat{g}_i$ be the densities of $\realD_i$ and $\wiggleD_i$ in polar coordinates. Claim~\ref{claim:jacobian} is a well-known consequence of multi-variable calculus. 

\begin{claim}
\label{claim:jacobian}
Let $\mathbf{x} \in \mathbb{R}^+_m$ and let $(r, \boldsymbol{\theta})$ be its polar transform. For any probability density $f(\mathbf{x})$ defined on $\mathbf{x} \in \mathbb{R}^+_m$, the corresponding density in polar coordinates is given by $g(r, \boldsymbol{\theta}) = J_m f(\mathbf{x}),$
where $J_m =r^{m-1} \prod_{i=1}^{m-2} \sin(\theta_i)^{m-1-j}$. In particular $J_2 = r$, $J_3 = r^2 \sin \theta_1$ and so on. 
\end{claim}

We use the following Lemma in all our upper bounds; it's a generalization of the fixed angle upper bound in the proof of Theorem~\ref{thm:angle theorem new}.

\begin{lemma}
\label{lemma:anglefocus}
For any pricing $p$, for any angles $\boldsymbol{\theta}_1, \dots, \boldsymbol{\theta}_n$ , where $\boldsymbol{\theta}_i = (\theta_{i,1},\dots,\theta_{i,m-1})$, 
\[ \textstyle \int_{\mathbf{r} \in [0,\infty)^{n} } \hat{g}(\mathbf{r} | \boldsymbol{\theta}_1, \dots, \boldsymbol{\theta}_n ) \sum_{i=1}^n p^*_i(\mathbf{r}, \boldsymbol{\theta}_1, \dots, \boldsymbol{\theta}_n ) d\mathbf{r} \leq \rev{ \wiggleD_{\boldsymbol{\theta}_1}, \dots, \wiggleD_{\boldsymbol{\theta}_n} }, \]
where $\wiggleD_{\boldsymbol{\theta}_i}$ is the single parameter distribution of buyer $i$ (perhaps correlated with the distribution of another buyer $j$), where a length vector $\mathbf{r}$ is drawn according to the density function $\hat{g} ( \mathbf{r} | \boldsymbol{\theta}_1, \dots, \boldsymbol{\theta}_n )$, and $\rev{\wiggleD_{\boldsymbol{\theta}}}$ is the revenue of the optimal auction that is truthful in the same sense as $p^*$.
\end{lemma}
\begin{proof}
The LHS is the revenue of $p$ conditioned on the angles being $\boldsymbol{\theta}$, i.e. the revenue of $p$ for the distribution where we draw $\mathbf{r}$ according to $\hat{g}\left( \mathbf{r} | \boldsymbol{\theta} \right)$, and buyer $i$ has value $r_i \trig_j(\boldsymbol{\theta}_i)$ for item $j$. The RHS is the maximum revenue that a truthful auction can extract from the same distribution.
\end{proof}

%% file: silly_model.tex
\section{Additive noise}\label{app:silly}

We provide almost matching upper and lower bounds for the model where $\realD$ is supported on $[0,v_{\max}]^2$ and we perturb $\mathbf{v}$ by adding a uniformly random vector from $[0,\delta v_{\max}]^2$. 

\begin{lemma}
$\rev{\wiggleD} \in O(\ln(1/\delta)) \brev{\wiggleD}$ 
\end{lemma}

\begin{proof}
Observe that for any $\wiggleD$, and any $x \geq \brev{\wiggleD}$, we have that $Pr[ \hat{v}_1 + \hat{v}_2 \geq x] \leq \frac{\brev{\wiggleD}}{x}$. Otherwise, if there exists $x^* \geq \brev{\wiggleD}$ such that $Pr[ \hat{v}_1 + \hat{v}_2 \geq x^*] > \frac{\brev{\wiggleD}}{x^*}$, setting a posted price of $x^*$ for the grand bundle would yield strictly more revenue than $\brev{\wiggleD}$; a contradiction.
\begin{align*}
\mathbb{E}[ \hat{v}_1 + \hat{v}_2 ] &= \int_{x=0}^2 Pr[ \hat{v}_1 + \hat{v}_2 \geq x ] dx \\
&=\int_{x = 0}^{\brev{\wiggleD}} Pr[ \hat{v}_1 + \hat{v}_2 \geq x ] dx  + \int_{x = \brev{\wiggleD}}^2 Pr[ \hat{v}_1 + \hat{v}_2 \geq x ] dx \\
&\leq \brev{\wiggleD} + \int_{x = \brev{\wiggleD}}^2 \frac{\brev{\wiggleD}}{x} dx \\
&= \brev{\wiggleD} + \brev{\wiggleD} \ln\left( \frac{2}{\brev{\wiggleD}} \right),
\end{align*}
which immediately implies that $\rev{\wiggleD} \leq \brev{\wiggleD} + \brev{\wiggleD} \ln\left( \frac{2}{\brev{\wiggleD}} \right)$. Furthermore, $\brev{\wiggleD} \geq \delta$, by simply selling the grand bundle for a price of $\delta$. The lemma follows.
\end{proof}

\begin{lemma}
There exists a bivariate distribution $\realD$  such that for its corresponding perturbed distribution $\wiggleD$, $\rev{\wiggleD} \in \Omega\left( \ln(\frac{1}{\delta})^{1/7} \right) \brev{\wiggleD}$.
\end{lemma}

\begin{proof}
First, we restate some known results that we will need.\footnote{These propositions appear in the arXiv version of~\cite{HartN13} uploaded on 22 April, 2013.}

Let $q^1, q^2, \dots$ be a finite or countably infinite sequence of points in $[0, 1]^k$. Define
\[
gap^n = q^n \cdot q^n - \max_{j < n} q^j \cdot q^n.
\]

\begin{proposition}[Proposition $5.1$ in~\cite{HartN13}]
For every (finite or countably infinite) sequence $q^1, q^2, \dots$ of points in $[0, 1]^k$ there exists a distribution $F$ on $\mathbb{R}^k_+$ such that $\brev{F} \leq 2k$, and $\rev{F} \geq \sum_{n} gap^n$. Specifically, for each $n$, let $M^n = (2k)^n/(\prod_{j=1}^n gap^j )$. Then, $F$ is the distribution that puts probability $1/M^n$ on the point $x^n = M^n q^n$, and puts the remaining probability on the point $0$.
\end{proposition}

\begin{proposition}[Proposition $6.1$ in~\cite{HartN13}]
There exists an infinite sequence of points $q^1, q^2, \dots$ of points in $[0,1]^2$ with $\| q^k \|_2 \leq 1$, such that taking for all $k$, $gap^k \in \Omega(k^{-6/7})$.
\end{proposition}

Combining these two results we can get a distribution $D$ over $n$ points in $[0,1]^2$, such that $\brev{D} \leq 4$ and $\rev{D} \geq \sum_{k=1}^n gap^k = \sum_{k=1}^n k^{-6/7} \in \Omega( n^{1/7} )$. The largest possible value for any of the two items is $M^n q^n \leq 2 M^n \leq  2 \frac{4^n}{\prod_{j=1}^n gap^j } \leq 4^{n+1}$.

Let $\realD$ be $D / 4^{n+1}$. Then, $\brev{\realD} \leq 4^{-n}$ and $\rev{\realD} \in \Omega( n^{1/7} / 4^{n} )$.
To complete the proof of the lemma, we find a $\delta$ that is large, but at the same time the gap between $\rev{\wiggleD}$ and $\brev{\wiggleD}$  remains sufficiently large as well. We use the following theorem:

\begin{theorem}[\cite{RubinsteinW15}]
Let $D$ be a distribution, and $D^+$ be a distribution that stochastically dominates $D$. Let $\Delta$ denote the random function $v^+ - v$, when couples $v^+$ and $v$ are sampled jointly from $D^+$ and $D$. Finally, let $\textsc{Val}(\Delta)$ be the expected welfare of VCG with buyers whose types are distributed according to $\Delta$. Then, for all such $D$, $D^{+}$, $\Delta$, and for every mechanism $M$ for $D$ and every $\epsilon > 0$, there exists a mechanism $M'$ for $D^{+}$ such that
\[
\rev{ M', D^{+} } \geq (1-\epsilon) \cdot \left( \rev{ M, D} - \frac{\textsc{Val}(\Delta)}{\epsilon} \right).
\]
\end{theorem}

We apply this theorem for $D = \realD$, and $D^{+} = \wiggleD$, which stochastically dominates $\realD$, and $M$ the optimal mechanism for $\realD$. Our noise model adds at most $\delta$ to the value of each item, therefore $\textsc{Val}(\Delta) \leq 2 \delta$. We set $\epsilon = \frac{1}{2}$. Then:
\begin{align*}
\rev{\wiggleD} &\geq \rev{ M', \wiggleD } \\
&\geq \frac{1}{2} \cdot \left( \rev{\realD} - 4\delta \right) \\
&\geq \frac{n^{1/7}}{4^{n+1}} - 2\delta.
\end{align*}
Finally, $\brev{\wiggleD} \leq \brev{\realD} + 2\delta \leq 4^{-n} + 2 \delta$.
Picking $n \in \Theta( \ln( \frac{1}{\delta} ) )$ completes the lemma.
\end{proof}

%% file: badmodels.tex
\section{Unilluminating Alternative Models}\label{app:badsmoothed}
In this section we briefly overview a few seemingly natural alternatives that unfortunately have an unilluminating analysis (similar to the additive model from Section~\ref{sec:notation}). \textbf{It is certainly a worthwhile direction for future research to identify different illuminating models}, but the point of this section is just to provide further evidence that doing so requires care. This section only contains sketches and no formal proofs. Both alternatives below propose some restrictions on the class of distributions $D$ considered, and ask what guarantees can be made subject to these restrictions.

\subsection{Lower-bounded Density}
Consider the following restriction on distributions: for all $\vec{v}$ in the support of $\mathcal{D}$, $f(\vec{v}) \in [\delta, 1/\delta]$. Let $K$ denote the support of $\mathcal{D}$. There are two cases to consider:

First, maybe $K$ is not convex. If $K$ is not required to be convex, then we can take exactly the construction of~\cite{BriestCKW15, HartN13} and turn each discrete point into a tiny uniform ball where the density is in the range $[\delta, 1/\delta]$, all without affecting the properties at all.

Second, maybe $K$ is required to be convex. If $K$ is convex, then consider $L = \max_{\vec{v} \in K}\{v_1 + v_2\}$. First, observe that at least $1/4$ of the volume of $K$ (not weighted by density in $\mathcal{D}$) lies above the line $v_1 + v_2 \geq L/2$. To see this, consider shrinking $K$ by a factor of $2$ in all dimensions, and moving it back to touch the line $v_1 + v_2 = L$. Then clearly this lies inside $K$, and has volume at least $\text{vol}(K)/4$. Now, observe that because the density lies in $[\delta, 1/\delta]$ everywhere, this means that the total \emph{density} of $\mathcal{D}$ above the line $v_1+v_2 \geq L/2$ is at least $\delta^2/4$. Therefore, by simply setting price $L/2$ on the grand bundle, we can generate revenue $L\delta^2/8$. The optimum revenue is clearly at most $L$, so this guarantees an $O(\delta^2)$ approximation. Like with the initial silly model, we simply don't learn anything from this analysis. All we learn is that if the densities are sufficiently bounded, then you can get a decent approximation by picking a specific type and targeting them while ignoring the rest. 

\subsection{Bounded away from Axes}
Consider the following restriction on distributions: with probability $1-\delta$, $\vec{v} \leftarrow \mathcal{D}$ satisfies $v_1 \geq \delta v_2 \geq \delta^2 v_1$. In other words, $\vec{v}$ is sufficiently bounded away from the axes with non-negligible probability. Unfortunately this buys us absolutely nothing. Consider the construction of~\cite{BriestCKW15, HartN13}, and the distribution $\mathcal{D}'$ which draws from their construction with probability $\delta$, and otherwise draw from $U([1,2]\times[1,2])$. Then the optimal revenue of this distribution is still infinite, and the revenue of any mechanism with menu complexity $C$ is at most $O(C)$, even though with high probability it is bounded away from the axes. 

The above construction in fact works for any property. Consider restricting to distributions such that with probability $1-\delta$, $\vec{v} \leftarrow \mathcal{D}$ has property X. Assuming that some $\vec{v}$ exists with property X, let $\mathcal{D}'$ draw from the~\cite{BriestCKW15, HartN13} construction with probability $\delta$, and draw $\vec{v}$ otherwise. Then the optimal revenue of this distribution is still infinite, yet the optimal revenue for any mechanism with menu complexity $C$ is at most $O(C(v_1+v_2))$.